\documentclass[preprint,12pt]{elsarticle}

\usepackage{amssymb}
\usepackage{ucs}
\usepackage{amsmath}
\usepackage{amsfonts}
\usepackage{amssymb,latexsym}
\usepackage{graphicx}
\usepackage{dsfont}
\usepackage{mathrsfs}
\usepackage{appendix}
\usepackage{amsthm}
\newtheorem{prop}{Proposition}
\newtheorem{lem}{Lemma}

\journal{Computational Statistics and Data Analysis}

\begin{document}

\begin{frontmatter}



\author[MM]{Matthieu Marbac}
\author[CB]{Christophe Biernacki}
\author[VV]{Vincent Vandewalle}

 \address[MM]{DGA \& Inria Lille \& University Lille 1}
 \address[CB]{University Lille 1 \& CNRS \& Inria Lille}
 \address[VV]{EA 2694 University Lille 2 \& Inria Lille}

\title{Finite mixture model of conditional dependencies modes to cluster categorical data.}

\address{}

\begin{abstract}
We propose a parsimonious extension of the classical latent class model to cluster categorical data by relaxing the class conditional independence assumption. Under this new mixture model, named Conditional Modes Model, variables are grouped into conditionally independent blocks. The corresponding block distribution is a parsimonious multinomial distribution where the few free parameters correspond to the most likely modality crossings, while the remaining probability mass is uniformly spread over the other modality crossings. Thus, the proposed model allows to bring out the intra-class dependency between variables and to summarize each class by a few characteristic modality crossings. The model selection is performed via a Metropolis-within-Gibbs sampler to overcome the computational intractability of the block structure search. As this approach involves the computation of the integrated complete-data likelihood, we propose a new method (exact for the continuous parameters and approximated for the discrete ones) which avoids the biases of the \textsc{bic} criterion  pointed out by our experiments. Finally, the parameters are only estimated for the best model via an \textsc{em} algorithm. The characteristics of the new model are illustrated on simulated data and on two biological data sets. These results strengthen the idea  that this simple model allows to reduce biases involved by the conditional independence assumption and gives meaningful parameters. Both applications were performed with the R package \texttt{CoModes}\footnote{Downloadable at https://r-forge.r-project.org/R/?group\_id=1809} where the proposed model is implemented.
\end{abstract}


\end{frontmatter}


\section{Introduction}
Clustering \citep{Jaj02} is an important tool for practitioners confronted with a complex data set. Indeed, this method allows to extract main information from data by grouping individuals into \emph{homogeneous} classes. This paper focuses on categorical variables clustering. Even if these variables are little informative, they are present in many different fields (biology, sociology, marketing...) because they are usually easily accessible.

\bigskip
Clustering methods can be split into two approaches: the \emph{geometrical} ones based on the distances between individuals and the \emph{probabilistic} ones which model the data generation. 

If these first approaches are generally simpler and faster than the others, they are sensitive to the used distance between individuals.  Geometrical approaches, to cluster categorical data, either define a metric in the initial variable space like the \emph{k-means} \citep{Hua05}, either compute their metric on the axes of the multiple correspondence analysis \citep{Cha10,Gui01}. Indeed, these approaches consider that the classes are \emph{homogeneous} when the distance between the individuals of the class and its center is small. However, lots of geometrical approaches can be interpreted as probabilistic ones \citep{Gov10} revealing probabilistic hidden assumptions made by the geometrical ones. Moreover, the probabilistic approaches allow to solve difficult questions, like the class number selection, in a rigorous mathematical framework.

The probabilistic approaches consider that the classes are \emph{homogeneous} when the individuals of a class are drawn from the same distribution. Thus, finite mixture models, which are the most classical of the probabilistic approaches, meet this objective by approaching the data distribution with a finite mixture of parametric distributions \citep{Mcl00}. In addition, the obtained partition is meaningful since each class is described by the parameters of the corresponding component.

\bigskip
The most widely used mixture model to cluster categorical data sets is the latent class model \citep{Goo74,Cel91,Bie10}, which assumes the conditional independence between variables. In this article, we refer to this model as \emph{Conditional Independence Model} (further mentioned as \textsc{cim}). In this model, as the classes are explicitly described by the probability of each modality for each variable, the interpretation is easy. Moreover, the sparsity caused by the conditional independence assumption is a great advantage since it circumvents the curse of dimensionality. In practice, this model obtains good results in lots of applications \citep{Han01,Reb06,Str06}. However,  applications \citep{Van09} can show that \textsc{cim} overestimates the class number when the conditional independence assumption is violated (see also our experiments presented in Section~\ref{sec::simul}). Furthermore, the larger is the number of variables, the higher is the risk to observe conditionally correlated variables in a data set, and consequently the higher is the risk to involve such biases by using \textsc{cim}.

Different models relax the conditional independence assumption. Among them, the \emph{multilevel latent class model} \citep{Ver03,Ver07} assumes that the conditional dependency between the observed variables can be explained by other unobserved variables. This model has connections with the approach modeling the intra-class dependencies by  using a latent continuous variable and a probit function  \citep{Qu96}. Recently proposed, the \emph{mixture of latent trait analyzers} \citep{Gol13,Bar11} is a good challenger for \textsc{cim} since it assumes that the categorical variable distribution depends on many latent variables: one categorical variable (the class) and many continuous latent variables (modeling the intra-class dependencies between the observed categorical variables). However, the parameters are hardly estimated directly, so the authors use a variational approach. Furthermore, the intra-class dependencies can be hardly interpretable by the practitioner, since the correlations are interpreted according to relationships with unobserved continuous variables.

The log-linear models  \citep{Agr02,Boc86} purpose is to model the individual log-probability by selecting interactions between variables. Thus, the most general mixture model is the \emph{log-linear mixture model} where all the kinds of interactions can be considered. It has been used for a long time \citep{Hag88}  and it obtains good results in many applications like the clustering of radiographic cross-diagnostics \citep{Esp89} or in a market segmentation \citep{Van09}. However this model family is huge and the model selection is a real challenge. In the literature, authors fix by advance the modeled interactions or they perform a deterministic search like the \emph{forward} method which is sub-optimal. Furthermore, the number of parameters increases with the conditional modality crossings, so there is an over-fitting risk and the interpretation becomes harder. 

\bigskip
In this paper, we propose a sparse mixture model relaxing the conditional independence assumption to overcome the biases caused by \textsc{cim}. This new model, named \emph{Conditional Modes Model} (refered in this article by \textsc{cmm}), groups the variables into \emph{conditionally independent blocks}, allowing to consider the main conditional dependencies. Such an idea was already proposed to cluster continuous and categorical data in the Multimix software \citep{Jor96,Hun99}. However, the specific distribution of the block that we adopt here is a multinomial per modes distribution which assumes that few modality crossings, named \emph{modes}, are characteristic and that the other ones follow a uniform distribution. Thus, the associated multinomial distribution is parsimonious, its free parameters being limited to the few parameters of the modes.

This simple mixture model (\textsc{cmm}) is a good challenger for the mixture model with conditional independence assumption (\textsc{cim}), since it preserves the sparsity and avoids many biases through modeling of the main conditional correlations. It can be also interpreted as a parsimonious version of the log-linear mixture model. Indeed, the repartition of the variables into blocks defines the considered interactions while the mode distribution into blocks defines a specific distribution for each interaction. Furthermore, resulting classes are meaningful since the intra-class dependencies are brought out at two complementary levels: the block variable interaction level and the associated mode interaction level (through locations and probabilities). Note that \textsc{cmm} is a comprehensive approach since it includes \textsc{cim} and a part of its parsimonious versions \citep{Cel91}.

For a fixed model (class number, repartition of the variables into blocks and mode numbers), the maximum likelihood estimate is obtained via an \textsc{em} algorithm. 
The model selection is performed via a Metropolis-within-Gibbs sampler generating a new block variable repartition into blocks and new mode numbers by a Metropolis-Hastings step. It is performed for a fixed number of classes and avoids combinatorial problems involved by the selection of the blocks of variables and of the mode number. It is based on the fact that the integrated complete-data likelihood, required for the acceptance probability computation of the Metropolis-Hastings inside the Gibbs sampler, is accessible and non ambiguous through weekly informative conjugate prior. Finally, this approach has two main advantages. It allows to reduce the bias of the \textsc{bic}-like approach (the overestimation of the number of modes by this approach is illustrated during our numerical experiments). Furthermore, it allows to perform an efficient model selection in a reasonable computational time since the parameters are only estimated for the unique selected model. Thus, this approach is a possible answer to the combinatorial model selection problem which is known to be a real challenge for a log-linear mixture model.

\bigskip
This paper is organized as follows. Section~\ref{sec::pres} presents the Conditional Modes Model. Section~\ref{sec::param} is devoted to maximum likelihood estimation via an \textsc{em} algorithm. Section~\ref{sec::model} presents the Metropolis-within-Gibbs sampler performing the model selection through the integrated complete-data likelihood. In Section~\ref{sec::simul}, we show that the proposed approach computing the integrated complete-data likelihood sharply reduces the biases of the \textsc{bic}-like approach and we numerically underline both the good behavior of the Metropolis-within-Gibbs sampler and the flexibility of \textsc{cmm} on simulated data. Section~\ref{sec::appli} presents two clusterings of biological data sets performed by the R package \texttt{CoModes}\footnote{Downloadable at https://r-forge.r-project.org/R/?group\_id=1809}. A conclusion is drawn and future extensions are discussed in Section~\ref{sec::concl}.

\section{Conditional modes model\label{sec::pres}}

\subsection{Conditional modes model framework}
Observations are described with \textsc{b} categorical variables $\textbf{x}=(\textit{\textbf{x}}^1,\ldots,\textit{\textbf{x}}^{\textsc{b}})$ using the complete disjunctive coding, where $\textit{\textbf{x}}^b$ has $m_b$ modalities. Let a partition $\boldsymbol{\sigma}=(\boldsymbol{\sigma}_1,\ldots,\boldsymbol{\sigma}_d)$ of $\{1,\ldots,\textsc{b}\}$ determining a repartition of variables in $d$ blocks. The $j$-th block is also denoted by 
\begin{equation}\label{data}
\textbf{x}^{j}=\{\textit{\textbf{x}}^b;b\in\boldsymbol{\sigma}_j\}.
\end{equation}
We still adopt the disjunctive coding $\text{x}^{jh}=1$ if the individual takes the modality crossing $h$ and $\text{x}^{jh}=0$ otherwise, where $h\in\{1,\ldots,\text{m}_j\}$. Thus, $\textbf{x}^{j}$ corresponds to a new categorical variable having $\text{m}_j$ modalities, $\text{m}_j$ being the number of modality crossings of the initial variables affected into the block $j$ determined by $\text{m}_j=\prod_{b\in\boldsymbol{\sigma}_j} m_b$. 

The proposed model, below refered as Conditional Modes Model (\textsc{cmm}), assumes that data arise independently from a mixture of $g$ components of \emph{independent blocks of initial variables}, the intra-block distribution being \emph{multinomial per modes}. In each class, the multinomial distribution of $\textbf{x}^{j}$ is also assumed to have few free parameters in comparison to the number of modalities $\text{m}_j$ and corresponding to the \emph{modes} of the distribution. More precisely, they are defined as the locations of the largest probabilities, while the other parameters are equal. A particular model is denoted by $\boldsymbol{\omega}=(g,\boldsymbol{\sigma},\boldsymbol{\ell})$ where $\boldsymbol{\ell}=(\boldsymbol{\ell}_1,\ldots,\boldsymbol{\ell}_{g})$ groups all the mode numbers with $\boldsymbol{\ell}_{k}=(\ell_{k1},\ldots,\ell_{kd})$, $\ell_{kj}$ being the number of modes of $\textbf{x}^{j}$ for the class $k$ $(0<\ell_{kj}<\text{m}_j)$.

\subsection{Probability distribution functions of the conditional modes model}
 Using $p(.;.)$ as a generic notation for the probability distribution function (pdf), for a known model $\boldsymbol{\omega}$, the \textsc{cmm}'s pdf can be written as
\begin{equation}
p(\textbf{x} ; \boldsymbol{\theta}, \boldsymbol{\omega})=\sum_{k=1}^g \pi_k p(\textbf{x} ; \boldsymbol{\alpha}_k, \boldsymbol{\sigma}, \boldsymbol{\ell}_k),
\end{equation}
where $\boldsymbol{\theta}=(\boldsymbol{\pi},\boldsymbol{\alpha})$ denotes the whole mixture parameters: $\boldsymbol{\pi}=(\pi_1,\ldots,\pi_g)$ is the vector of class proportions with $0<\pi_k\leq 1$ and $\sum_{k=1}^g \pi_k=1$, and $\boldsymbol{\alpha}=(\boldsymbol{\alpha}_1,\ldots,\boldsymbol{\alpha}_g)$ groups the parameters of the multinomial distributions with $\boldsymbol{\alpha}_k=(\boldsymbol{\alpha}_{k1},\ldots,\boldsymbol{\alpha}_{kd})$ and $\boldsymbol{\alpha}_{kj}=(\alpha_{kj1},\ldots,\alpha_{kj\text{m}_j})$,  $\alpha_{kjh}$ being the probability that $\text{x}^{jh}=1$ conditionally to the component $k$.

 We also define the mapping $\tau_{kj}$ from $\{1,\ldots,\text{m}_j\} $ to $\{1,\ldots,\text{m}_j\} $ 
which orders the modalities of $\textbf{x}^j$ by decreasing values of the probabilities $\alpha_{kjh}$. For instance, $\tau_{kj}(1)$ gives the modality of $\textbf{x}^j$ having the largest probability $\alpha_{kjh}$. By using the shorter notation $\alpha_{kj(h)}=\alpha_{kj\tau_{kj}(h)}$, we have $\alpha_{kj(h)}\geq\alpha_{kj(h+1)}$ $(1\leq h<\text{m}_j)$. Furthermore, since the multinomial distribution of $\textbf{x}^j$ has $\ell_{kj}$ modes then $\boldsymbol{\alpha}_{kj}$ is defined in the constrained simplex $S(\ell_{kj},\text{m}_j)$ where
\begin{equation}
S(\ell_{kj},\text{m}_j)=\left\{\boldsymbol{\alpha}_{kj}:\; \sum_{h=1}^{\text{m}_j} \alpha_{kjh}=1,\; \alpha_{kj(\ell_{kj}+1)}=\ldots=\alpha_{kj(\text{m}_j)}\right\}. \label{simplex}
\end{equation}
In other words, uniformity holds for non-mode modalities.
By using the other shorter notation $\text{x}^{j(h)_k}=\text{x}^{j\tau_{kj}(h)}$, the conditional independence assumption between blocks involves the following pdf for the component $k$
\begin{equation}
p(\textbf{x} ; \boldsymbol{\alpha}_k, \boldsymbol{\sigma}, \boldsymbol{\ell}_k)=\prod_{j=1}^d \prod_{h=1}^{\text{m}_j} \left(\alpha_{kjh} \right)^{\text{x}^{jh}}=\prod_{j=1}^d \prod_{h=1}^{\text{m}_j} \left(\alpha_{kj(h)} \right)^{\text{x}^{j(h)_k}}.
\end{equation}
Thus, \textsc{cim} is included in \textsc{ccm}, since the conditional independence assumption between the initial variables is defined by putting $d=\textsc{b}$ and $\ell_{kj}=m_{j}-1$. Indeed, in such case, each variable built a block, so $\boldsymbol{\sigma}=(\{1\},\ldots,\{\textsc{b}\})$.

\subsection{Conditional modes model characteristics}
\textsc{cmm} is meaningful with its two levels of interpretation. Firstly, the intra-class dependencies of variables (equal between classes) are brought out by the repartition of the variables into blocks given by $\boldsymbol{\sigma}$. Secondly, the intra-class and intra-block dependencies of modalities (possibly different between classes) are summarized by the modes (locations and probabilities). A shorter summary for each distribution is also available by using the following compact terms, defined on $[0,1]$ and respectively reflecting the \emph{complexity} and the \emph{strength} of the intra-class and intra-block dependencies
\begin{equation}
\kappa_{kj}=\frac{\ell_{kj}}{\text{m}_j-1} \text{ and } \rho_{kj}=\sum_{h=1}^{\ell_{kj}}\alpha_{kj(h)}. \label{summ}
\end{equation}
For instance, the smaller is $\kappa_{kj}$ and the larger is $\rho_{kj}$, the more  massed in few characteristic modality crossings is the distribution, since the modes are interpreted as an over-contribution at the uniform distribution among all the modality crossings.

Note that the repartition of the variables into conditionally independent blocks identical between classes assures the model generic identifiability. Indeed, with this constraint, the results of \citep{All09} can be applied to prove the generic identifiability of the \textsc{cmm} (details are given in \ref{ident}). Despite the constraint of the same repartition of the variables into blocks for all the classes, the model stays flexible because of the specific block distribution.

The main idea of the former sparse versions of \textsc{cim} (classical conditional independence model) proposed by \cite{Cel91} is to estimate only one mode for each multinomial distribution of the initial variable. Different constraints of equality are then added between the variables and/or classes. In fact, many of these models are included in the model family of \textsc{cmm} by putting $d=\textsc{b}$ and $\ell_{kj}=1$. In addition, as \textsc{cmm} needs $\nu=(g-1) + \sum_{k=1}^g \sum_{j=1}^{d} \ell_{kj}$ parameters, models of \textsc{cmm}'s family can need less parameters than \textsc{cim}---having $(g-1) + g \times\sum_{b=1}^{\textsc{b}} (m_{b}-1)$ parameters---although it takes into account the conditional dependencies.

\subsection{New parametrization of the block distribution} \label{sec::reparam}
The parsimonious versions of \textsc{cim} introduced by \cite{Cel91} are meaningful since each multinomial is written with two parameters: one discrete giving the location of the mode of the distribution and one continuous giving its probability. By using the same idea, we propose a new parametrization of the block distribution,  denoted by $(\boldsymbol{\delta}_{kj},\boldsymbol{a}_{kj})$, which facilitates the interpretation and the writing of the prior and posterior distributions of the block parameters (see Section~\ref{sec::model}). The discrete parameter $\boldsymbol{\delta}_{kj}=\{\delta_{kjh};h=1,\ldots,\ell_{kj}\}$ determines the mode locations, since $\delta_{kjh}$  indicates the modality crossing where the mode $h$ is located, with $\delta_{kjh}\neq\delta_{kjh'}$ if $h\neq h'$ and $\delta_{kjh} \in \{1,\ldots,\text{m}_j\}$. The continuous parameter $\boldsymbol{a}_{kj}=(a_{kjh};h=1,\ldots,\ell_{kj}+1)$ determines the probability mass of the $\ell_{kj}$ modes by its first $\ell_{kj}$ elements ($a_{kjh}$ with $h=1,\ldots,\ell_{kj})$ and the probability mass of the non-mode by its last element ($a_{kj\ell_{kj}+1}$). This parameter is defined on the truncated simplex involving $a_{kjh}\geq \frac{a_{kj\ell_{kj}+1}}{\text{m}_j -\ell_{kj}}$ $(h\leq \ell_{kj})$. The parameter $\boldsymbol{\alpha}_{kj}$ and the couple $(\boldsymbol{\delta}_{kj},\boldsymbol{a}_{kj})$ are linked by
\begin{equation} \label{dirac}
\alpha_{kjh} = \left\{ \begin{array}{rl}
a_{kjh'} & \text{if } \exists h'\text{ such that } \delta_{kjh'}=h \\
\frac{a_{kj\ell_{kj}+1}}{\text{m}_j - \ell_{kj}} & \text{otherwise.}
\end{array}
\right.
\end{equation}

\section{Maximum likelihood estimate \label{sec::param}}
The whole data set consisting of $n$ independent and identically distributed individuals is denoted by $\textbf{X} = (\textbf{x}_{1},\ldots,\textbf{x}_{n} )$. Remark that $\textbf{X}$ denotes the whole observed sample and not a random variable. The observed-data log-likelihood of \textsc{cmm} is also defined as
	\begin{equation}
L(\boldsymbol{\theta}; \mathbf{X},\boldsymbol{\omega})=\sum_{i=1}^{n} \ln   p(\textbf{x}_{i}; \boldsymbol{\theta},\boldsymbol{\omega}). \label{vrai}
	\end{equation}
Since we use \textsc{cmm} to cluster, the indicator vector of the $g$ classes denoted by $\textbf{Z}=(\textbf{z}_{i};i=1,\ldots,n)$ with $\textbf{z}_i=(\text{z}_{i1},\ldots,\text{z}_{ig})$ where $\text{z}_{ik}=1$ if the individual $\textbf{x}_i$ arises from the class $k$ and $\text{z}_{ik}=0$ otherwise, is considered as a missing variable. The complete-data log-likelihood of \textsc{cmm} is then defined by
\begin{align}
L(\boldsymbol{\theta}; \textbf{X},\textbf{Z},\boldsymbol{\omega})&=\sum_{i=1}^n\sum_{k=1}^g \text{z}_{ik} \ln \big( \pi_k p(\textbf{x}_{i}; \boldsymbol{\alpha}_k,\boldsymbol{\sigma},\boldsymbol{\ell}_k)\big).
\end{align}

For the mixture models, the direct optimization on $\boldsymbol{\theta}$ to obtain the maximum likelihood estimate (\textsc{mle}), denoted by $\hat{\boldsymbol{\theta}}$, involves solving equations having no analytical solution. So, we perform the parameter's estimation via an \textsc{em} algorithm \citep{Dem77,Mcl97}, which is often simple and efficient for the missing data situation. It is an iterative algorithm which alternates between two steps: the computation of the complete-data log-likelihood conditional expectation (\textsc{e} step) and its maximization (\textsc{m} step). 
At the iteration $[r]$, this algorithm is written as:\bigskip\\
\textbf{E step:} conditional probabilities computation \begin{equation} 
t_{ik}(\boldsymbol{\theta}^{[r]})=\frac{\pi_{k}^{[r]}p(\textbf{x}_i; \boldsymbol{\alpha}^{[r]}_k,\boldsymbol{\sigma},\boldsymbol{\ell}_k)}{\sum_{k'=1}^{g} \pi_{k'}^{[r]} p(\textbf{x}_i; \boldsymbol{\alpha}^{[r]}_{k'},\boldsymbol{\sigma},\boldsymbol{\ell}_{k'}) }.\label{Estep1}\end{equation}
\textbf{M step:} maximization of the complete-data log-likelihood 
\begin{equation}
\pi^{[r+1]}_k=\frac{n_k^{[r]}}{n}
\text{ and }
\alpha_{kj(h)}^{[r+1]}=\left\{ \begin{array}{rl}
\frac{n_{kj(h)}^{[r]}}{n_k^{[r]}} & \text{ if } (1\leq h \leq \ell_{kj}) \\
\frac{1-\sum_{h'=1}^{\ell_{kj}} \alpha_{kj(h')}^{[r+1]}}{\text{m}_j - \ell_{kj}} & \text{ otherwise,}
\end{array}
\right.
\label{Mstep1}
\end{equation}
by using the  notations  $n_k^{[r]}=\sum_{i=1}^nt_{ik}(\boldsymbol{\theta}^{[r]})$ and  $n_{kjh}^{[r+1]}=\sum_{i=1}^n t_{ik}(\boldsymbol{\theta}^{[r]}) \text{x}_{i}^{jh}$. Note that, at the M step of iteration $[r]$, the function $\tau_{kj}$ is redefined as the decreasing ordering function of the $n_{kjh}^{[r+1]}$ and allows to define $n_{kj(h)}^{[r+1]}$ with $n_{kj(h)}^{[r+1]}\geq n_{kj(h+1)}^{[r+1]}$.

\section{Model selection via Metropolis-within-Gibbs sampler} \label{sec::model}
The aim is to obtain the model $\hat{\boldsymbol{\omega}}$ having the largest posterior probability. We assume that $p(g)=\frac{1}{g_{\max}}$ for $g=1,\ldots,g_{\max}$ and that $p(\boldsymbol{\sigma})$ (remind that $g$ and $\boldsymbol{\sigma}$ are independent) and $p(\boldsymbol{\ell}|g,\boldsymbol{\sigma})$ follow uniform distributions. 
Let the $g_{\max}$ models denoted by $\boldsymbol{\omega}^{(g)}=(g,\boldsymbol{\sigma}^{(g)},\boldsymbol{\ell}^{(g)})$, for $g=1,\ldots,g_{\max}$, where $(\boldsymbol{\sigma}^{(g)},\boldsymbol{\ell}^{(g)}) =\underset{\boldsymbol{\sigma},\boldsymbol{\ell}}{\operatorname{argmax }}\; p(\boldsymbol{\sigma},\boldsymbol{\ell} |\textbf{X},g)$. The best model is $\underset{g}{\operatorname{argmax }}\; p(\boldsymbol{\omega}^{(g)}|\textbf{X})$ and it is found by applying the \textsc{bic} approximation among those $g_{\max}$ selected models. 
However, an exhaustive search strategy is not feasible for two correlated reasons: firstly, the number of couples $(\boldsymbol{\sigma},\boldsymbol{\ell})$ can be excessively huge, and,  secondly, the estimation of the \textsc{mle} for each of them is an unnecessary waste of time computing. A Metropolis-within-Gibbs sampler strategy overcomes these two drawbacks at the same time, as we now describe.\bigskip

For a fix value of $g$, the couple $(\boldsymbol{\sigma}^{(g)},\boldsymbol{\ell}^{(g)}) $ is estimated by the following Metropolis-within-Gibbs sampler \citep{Rob04} having $p(\boldsymbol{\sigma},\boldsymbol{\ell} |g,\textbf{X})$ as stationary distribution and whose the iteration $[s]$  is written as
\begin{align}
\boldsymbol{\theta}^{[s+1]} &\sim p(\boldsymbol{\theta} | \boldsymbol{\omega}^{[s]}, \textbf{X},\textbf{Z}^{[s]}) \label{s_param} \\
\textbf{Z}^{[s+1]} &\sim p(\textbf{Z} | \boldsymbol{\omega}^{[s]},\textbf{X},\boldsymbol{\theta}^{[s+1]}) \\
(\boldsymbol{\sigma}^{[s+1]},\boldsymbol{\ell}^{[s+1]}) &\sim p(\boldsymbol{\sigma},\boldsymbol{\ell}|\boldsymbol{\omega}^{[s]},\textbf{X},\textbf{Z}^{[s+1]}), \label{s_model}
\end{align} 
where $\boldsymbol{\omega}^{[s]}=(g,\boldsymbol{\sigma}^{[s]},\boldsymbol{\ell}^{[s]})$. As the observed data are independent, the full conditional distribution of $\textbf{Z}$ is classical and is written as
\begin{equation}
p(\textbf{Z}| \boldsymbol{\omega},\textbf{X},\boldsymbol{\theta})=\prod_{i=1}^n p(\textbf{z}_i|\boldsymbol{\omega},\textbf{x}_i,\boldsymbol{\theta})
\text{ with }
 p(\textbf{z}_i|\boldsymbol{\omega},\textbf{x}_i,\boldsymbol{\theta})=\prod_{k=1}^g(t_{ik}(\boldsymbol{\theta}))^{\text{z}_{ik}}. \label{cl}
\end{equation}
In this section, we firstly detail the full conditional distributions sampling the parameters (denoted by \emph{instrumental elements}) by using block parametrization given in Section~\ref{sec::reparam}, and we secondly detail the sampling of $(\boldsymbol{\sigma},\boldsymbol{\ell})$ (considered as the \emph{interest elements}).

\subsection{Sampling of the instrumental elements} \label{estimbayes}
We now detail the sampling according to $ p(\boldsymbol{\theta} | \boldsymbol{\omega}^{[s]}, \textbf{X},\textbf{Z}^{[s]}) $ defined in \eqref{s_param}.
\paragraph{Prior assumption} We assume the independence \emph{a priori} between the class proportions and the block distribution parameters, involving that the prior of the whole parameter is written as
\begin{equation}
p(\boldsymbol{\theta}|\boldsymbol{\omega})=p(\boldsymbol{\pi}|\boldsymbol{\omega})\prod_{k=1}^g \prod_{j=1}^{d}p(\boldsymbol{\alpha}_{kj}|\boldsymbol{\omega}).
\end{equation}
Note that this property of conditional independence is preserved for the distribution of $\boldsymbol{\theta}$ conditionally on $(\boldsymbol{\omega},\textbf{X},\textbf{Z})$, thus
 \begin{equation}
  p(\boldsymbol{\theta}|\boldsymbol{\omega},\textbf{X},\textbf{Z})=p(\boldsymbol{\pi}|\boldsymbol{\omega},\textbf{X},\textbf{Z})\prod_{k=1}^g \prod_{j=1}^{d}p(\boldsymbol{\alpha}_{kj}|\boldsymbol{\omega},\textbf{X},\textbf{Z}).
 \end{equation}

\paragraph{Prior and posterior distributions of $\boldsymbol{\pi}$}
The Jeffreys non informative prior distribution, for a multinomial, is a conjugate Dirichlet distribution \citep{Rob07}. So, the prior and the posterior distributions of $\boldsymbol{\pi}$ \citep{Bie10} are respectively defined by
\begin{equation}
\boldsymbol{\pi}|\boldsymbol{\omega} \sim \mathcal{D}_g\Big(\frac{1}{2},\ldots,\frac{1}{2}\Big) \text{ and }
\boldsymbol{\pi}|\boldsymbol{\omega},\textbf{X},\textbf{Z} \sim \mathcal{D}_g\Big(\frac{1}{2}+\text{n}_1,\ldots,\frac{1}{2}+\text{n}_g\Big), \label{prop}
\end{equation}
where  $\text{n}_k=\sum_{i=1}^n\text{z}_{ik}$ (not equal to $n_k^{[r]}$).

\paragraph{Prior distribution of $\boldsymbol{\alpha}_{kj}$}
We now use the parametrization of the block distribution $(\boldsymbol{\delta}_{kj}, \boldsymbol{a}_{kj})$ (defined in Section~\ref{sec::reparam}). We assume the independence between the prior of $ \boldsymbol{\delta}_{kj}$  and of $\boldsymbol{a}_{kj} $, so 
\begin{align}
p(\boldsymbol{\alpha}_{kj}| \boldsymbol{\omega})=p(\boldsymbol{\delta}_{kj} | \boldsymbol{\omega})p(\boldsymbol{a}_{kj} | \boldsymbol{\omega}).
\end{align}
We use a uniform distribution among all the mode locations and a conjugate  truncated Dirichlet distribution\footnote{
$p(\boldsymbol{a}_{kj} |\boldsymbol{\omega})\propto \prod_{h=1}^{\ell_{kj}+1} (a_{kjh})^{\gamma_{kjh}-1} \mathds{1}_{ \left\{a_{kjh}\geq \frac{a_{kj\ell_{kj}+1}}{\text{m}_j -\ell_{kj}}\right\}}$.} as prior of $\boldsymbol{a}_{kj}$, so
\begin{equation}
p(\boldsymbol{\delta}_{kj}|\boldsymbol{\omega}) = \binom{\text{m}_j}{\ell_{kj}} ^{-1}
\text{ and }
\boldsymbol{a}_{kj}|\boldsymbol{\omega} \sim D_{\ell_{kj}+1}^t \Big(\gamma_{kj1},\ldots,\gamma_{kj\ell_{kj}+1};\text{m}_j\Big), \label{prior}
\end{equation}
where the $\gamma_{kjh}$ are the parameters of the truncated Dirichlet distribution. In \ref{proofs}, we justify why we now fix $\gamma_{kjh}=1$. The proposed prior is also weakly informative since it is an uniform distribution.

\paragraph{Posterior distribution of  $\boldsymbol{\alpha}_{kj}$}
The posterior distribution of $\boldsymbol{\alpha}_{kj}$ is written as
\begin{equation} \label{param}
p(\boldsymbol{\alpha}_{kj} | \boldsymbol{\omega},\textbf{X},\textbf{Z})=p(\boldsymbol{\delta}_{kj} | \boldsymbol{\omega},\textbf{X},\textbf{Z})p(\boldsymbol{a}_{kj} | \boldsymbol{\omega},\boldsymbol{\delta}_{kj},\textbf{X},\textbf{Z}). 
\end{equation}
The distribution of $\boldsymbol{\delta}_{kj}| \boldsymbol{\omega},\textbf{X},\textbf{Z}$ is a multinomial one with too many values to be computable.
Let the set $\tilde{\boldsymbol{\delta}}_{kj}=\{\tilde{\delta}_{kjh};h=1,\ldots,\ell_{kj}\}$ containing the indices of the $\ell_{kj}$ largest values of $\text{n}_{kjh}=\sum_{i=1}^{n}\text{z}_{ik}\text{x}_i^{jh}$ ordered such as
\begin{equation}
\forall h\in \{1,\ldots,\ell_{kj}-1\}, \quad \text{n}_{kj\tilde{\delta}_{kjh}}\geq \text{n}_{kj\tilde{\delta}_{kjh+1}}.
\end{equation}
We assume that the difference between the mode probabilities and the non-mode probabilities are significant. So, we can approximate the full conditional distribution of $\boldsymbol{\delta}_{kj}$ by a Dirac in $\tilde{\boldsymbol{\delta}}_{kj}$. This approximation is strengthened by the fast convergence speed of the discrete parameters \citep{Cho12}. Concerning now $\boldsymbol{a}_{kj}$, as its prior is conjugated, its conditional distribution is explicitly defined as
\begin{equation}
\boldsymbol{a}_{kj} | \boldsymbol{\omega},\boldsymbol{\delta}_{kj},\textbf{X},\textbf{Z} \sim \mathcal{D}_{\ell_{kj}+1}^{t}\Big(1 + \text{n}_{kj(1)},\ldots,1 + \text{n}_{kj(\ell_{kj})},1 + \bar{\text{n}}_{kj}^{\ell_{kj}}; \text{m}_j \Big),
\end{equation}
where $\text{n}_{kj(h)}$ is the $h$th larger value of the set $\{\text{n}_{kjh};h=1,\ldots,\text{m}_j\}$ and $\bar{\text{n}}_{kj}^{\ell_{kj}}=\text{n}_k - \sum_{h=1}^{\ell_{kj}} \text{n}_{kj(h)}$.

\subsection{Sampling of a new model according to $p(\boldsymbol{\sigma},\boldsymbol{\ell}|\boldsymbol{\omega}^{[s]},\textbf{X},\textbf{Z}^{[s+1]})$} \label{ar}
The sampling of $\boldsymbol{\omega}^{[s+1]}=(g,\boldsymbol{\sigma}^{[s+1]},\boldsymbol{\ell}^{[s+1]})$ according to Equation~\eqref{s_model} is performed in two steps. Firstly, a new repartition of the variables into blocks and the mode number of the modified blocks, respectively denoted by $ \boldsymbol{\sigma}^{[s+1]}$ and $\boldsymbol{\ell}^{[s+1/2]}$, are sampled by one iteration of a Metropolis-Hastings algorithm. Secondly, the mode number of each block is sampled by one \textsc{mcmc} iteration. Thus, the sampling of $\boldsymbol{\omega}^{[s+1]}$ is decomposed into the two following steps
\begin{align}
(\boldsymbol{\sigma}^{[s+1]},\boldsymbol{\ell}^{[s+1/2]})  &\sim p(\boldsymbol{\sigma},\boldsymbol{\ell} |\boldsymbol{\omega}^{[s]},\textbf{X},\textbf{Z}^{[s+1]}) \label{mo1}\\
 \boldsymbol{\ell}^{[s+1]} &\sim p(\boldsymbol{\ell} |\boldsymbol{\omega}^{[s+1/2]},\textbf{X},\textbf{Z}^{[s+1]}), \label{mo2}
\end{align}
where $\boldsymbol{\omega}^{[s+1/2]}=(g,\boldsymbol{\sigma}^{[s+1]},\boldsymbol{\ell}^{[s+1/2]})$. Thus, this chain has $p(\boldsymbol{\sigma},\boldsymbol{\ell}|g,\textbf{X},\textbf{Z}^{[s+1]})$ as stationary distribution.

\subsubsection{Metropolis-Hastings algorithm to sample $\boldsymbol{\omega}^{[s+1/2]}$}
The sampling of $\boldsymbol{\omega}^{[s+1/2]}$ is performed by one iteration of the Metropolis-Hastings algorithm divided into two steps. Firstly, the proposal distribution $q(.;\boldsymbol{\omega}^{[s]})$ generates a candidate $\boldsymbol{\omega}^{\star}=(g,\boldsymbol{\sigma}^{\star},\boldsymbol{\ell}^{\star})$. Secondly $\boldsymbol{\omega}^{[s+1]}$ is sampled according to the acceptance probability $\mu^{[s]}$ defined by
\begin{equation} \label{proba}
\mu^{[s]}= 1 \wedge \frac{ 
p(\textbf{X},\textbf{Z}^{[s]}|\boldsymbol{\omega}^{\star}) }{
p(\textbf{X},\textbf{Z}^{[s]}|\boldsymbol{\omega}^{[s]})}
\frac{
q(\boldsymbol{\omega}^{[s]};\boldsymbol{\omega}^{\star} )
}{
 q(\boldsymbol{\omega}^{\star};\boldsymbol{\omega}^{[s]} )
}.
\end{equation}
The computation of $\mu^{[s]}$ involves to compute the integrated complete-data likelihood. In Section~\ref{comp}, we described how to solve this problem without using the biased \textsc{bic} approximation or using too much time computing \textsc{mcmc} methods. The sampling of $\boldsymbol{\omega}^{[s+1/2]}$ is written as
\begin{align}
\boldsymbol{\omega}^{\star} &\sim q(.;\boldsymbol{\omega}^{[s]}) \\
\boldsymbol{\omega}^{[s+1/2]}&=\left\{ \begin{array}{rl}
\boldsymbol{\omega}^{\star} & \text{ with a probability } \mu^{[s]}\\
\boldsymbol{\omega}^{[s]} & \text{ otherwise.}
\end{array}
\right. \nonumber
\end{align}
The proposal distribution $ q(.;\boldsymbol{\omega}^{[s]})$ samples $\boldsymbol{\omega}^{\star}$ in two steps. The first step changes the block affectation of one variable. In practice, $\boldsymbol{\sigma}^{\star}$ is uniformly sampled in $V(\boldsymbol{\sigma}^{[s]})=\{\boldsymbol{\sigma}:\; \exists ! b \text{ as } b\in \boldsymbol{\sigma}_j^{[s]} \text{ and } b \notin \boldsymbol{\sigma}_j\}$. The second step  uniformly samples the mode numbers among all its possible values for the modified blocks while $\ell^{\star}_{kj}=\ell^{[s]}_{kj}$ for non-modified blocks (i.e. $j$ as $\boldsymbol{\sigma}_j^{[s]}=\boldsymbol{\sigma}_j^{\star}$).

\subsubsection{MCMC algorithm to sample $\boldsymbol{\ell}^{[s+1]}$}
This step allows to increase or decrease the mode number of each block by one at each iteration. So, $\ell_{kj}^{[s+1]}$ is sampled according to  $p(\ell_{kj} | \boldsymbol{\omega}^{[s+1/2]},\textbf{X},\textbf{Z}^{[s]})$ defined by
\begin{align}
p(\ell_{kj} | \boldsymbol{\omega}^{[s+1/2]},\textbf{X},\textbf{Z}^{[s+1]}) \propto 
\left\{ \begin{array}{rl}
p(\textbf{X}^{j}|\textbf{Z}^{[s+1]},\ell_{kj}) & \text{ if } |\ell_{kj}-\ell_{kj}^{[s+1/2]}|<2\\& \text{ and } \ell_{kj} \notin \{0,\text{m}_j\}. \\
0 & \text{ otherwise.}
\end{array}
\right. \label{chgt_ell}
\end{align}
Thus, this algorithm needs the value of  $p(\textbf{X}^{j}|\textbf{Z},\ell_{kj})$  defined by
\begin{equation}
p(\textbf{X}^{j}|\textbf{Z},\ell_{kj}) = \int_{S(\ell_{kj},\text{m}_j)} \prod_{j=1}^{\text{m}_j} (\alpha_{kjh})^{\text{n}_{kjh}} d\boldsymbol{\alpha}_{kj}.
\end{equation}
That we have to detail now.

\subsubsection{The integrated complete-data likelihood} \label{comp}
The integrated complete-data likelihood is defined as
\begin{equation}
p(\textbf{X},\textbf{Z}|\boldsymbol{\omega})=p(\textbf{Z}|\boldsymbol{\omega})\prod_{k=1}^g \prod_{j=1}^d p(\textbf{X}^{j}|\textbf{Z},\ell_{kj}),
\end{equation}
where $\textbf{X}^{j}=(\textbf{X}_i^{j};i=1,\ldots,n)$. Note that the quantities $p(\textbf{X},\textbf{Z}|\boldsymbol{\omega})$ and $p(\textbf{X}^{j}|\textbf{Z},\ell_{kj})$ are respectively needed to compute the acceptance probability of the Metropolis-Hastings algorithm (see Equation~\eqref{proba}) and to sample the number of modes (see Equation~\eqref{chgt_ell}) and can be evaluated by \textsc{bic}-like approximations \citep{Sch78,Leb06}. For instance, the integrated complete-data likelihood is approximated by
\begin{equation}
\ln p(\textbf{X},\textbf{Z}|\boldsymbol{\omega})=\ln p(\textbf{X},\textbf{Z}| \boldsymbol{\theta}^{\star}, \boldsymbol{\omega}) - \frac{\nu}{2} \ln n + \mathcal{O}(1),
\end{equation}
$\boldsymbol{\theta}^{\star}$ begin the maximum complete-data likelihood estimate. However, this kind of approximation is only asymptotically true and over-estimates the mode numbers (see Section~\ref{sec::simul1}). As $\textbf{Z}|\boldsymbol{\omega}$ follows a uniform distribution among all the possible partitions, we propose to compute each $p(\textbf{X}^{j}|\textbf{Z},\ell_{kj})$ to obtain $p(\textbf{X},\textbf{Z}|\boldsymbol{\omega})$. This computation is not easy since $\boldsymbol{\alpha}_{kj}$ is defined on $S(\ell_{kj};\text{m}_j)$ and not on the whole simplex of size $\ell_{kj}$ (except when $\ell_{kj}=\text{m}_j-1$, in such case we can use the approach of \textsc{cim} \citep{Bie10}). An explicit formula is given in the following proposition whose the proof is given in \ref{proofs} by performing an exact computation of the integral over the continuous parameters and an approximation on the discrete ones.

\begin{prop}
The integrated complete-data likelihood is approximated, by neglecting the sum over the discrete parameters of the modes locations and by performing the exact computation on the continuous parameters, by
\begin{equation}
p(\textbf{X}^{j}|\textbf{Z},\ell_{kj}) \approx \left( \frac{1}{m_j-\ell_{kj}} \right)^{\bar{\text{n}}_{kj}^{\ell_{kj}}}  \prod_{h=1}^{\ell_{kj}} \frac{Bi\left(\frac{1}{\text{m}_j-h+1};\text{n}_{kj(h)}+1;\bar{\text{n}}_{kj}^{h}+1\right)}{\text{m}_j- h},
\end{equation}
where $Bi(x;a,b)=B(1;a,b)-B(x;a,b)$, $B(x;a,b)$ being the incomplete beta function defined by $B(x;a,b)=\int_0^x w^a(1-w)^b dw$.
\end{prop}
From the previous expression, its is straightforward to obtain $p(\textbf{X},\textbf{Z}|\boldsymbol{\omega})$.

\section{Simulations \label{sec::simul}}
\subsection{Integrated complete-data likelihood: comparison of both approaches}\label{sec::simul1}
\paragraph{Aim} 
During this experiment, we highlight the biases of the \textsc{bic} criterion for the selection of the number of modes and the gain provided by the proposed computation of the integrated complete-data likelihood.

\paragraph{Data generation}
As we want to compare both approaches for the selection of the number of modes, we simulate samples composed by $n$ i.i.d individuals arisen from a multinomial distribution per modes $\mathcal{M}_s(r,r,r,\frac{1-3r}{s-3},\ldots,\frac{1-3r}{s-3})$ with $s$ modalities  and three modes having a probability $r$. For different sizes of sample, $10^5$ samples are generated with different values of $(r,s)$.

\paragraph{Results}
 Figure~\ref{res::simul} gives a comparison between the proposed approach and the \textsc{bic}-like approximation for the selection of the number of modes. The proposed criterion obtains best results than the \textsc{bic} criterion in the four studied situations for the large size of sample. Furthermore, it allows to never overestimates the mode number. Finally, its variability is smaller than the \textsc{bic} criterion one. We enter now into more specific comments.
\begin{figure}[h!]
   \begin{minipage}{0.49\textwidth}
      \centering \includegraphics[scale=0.38]{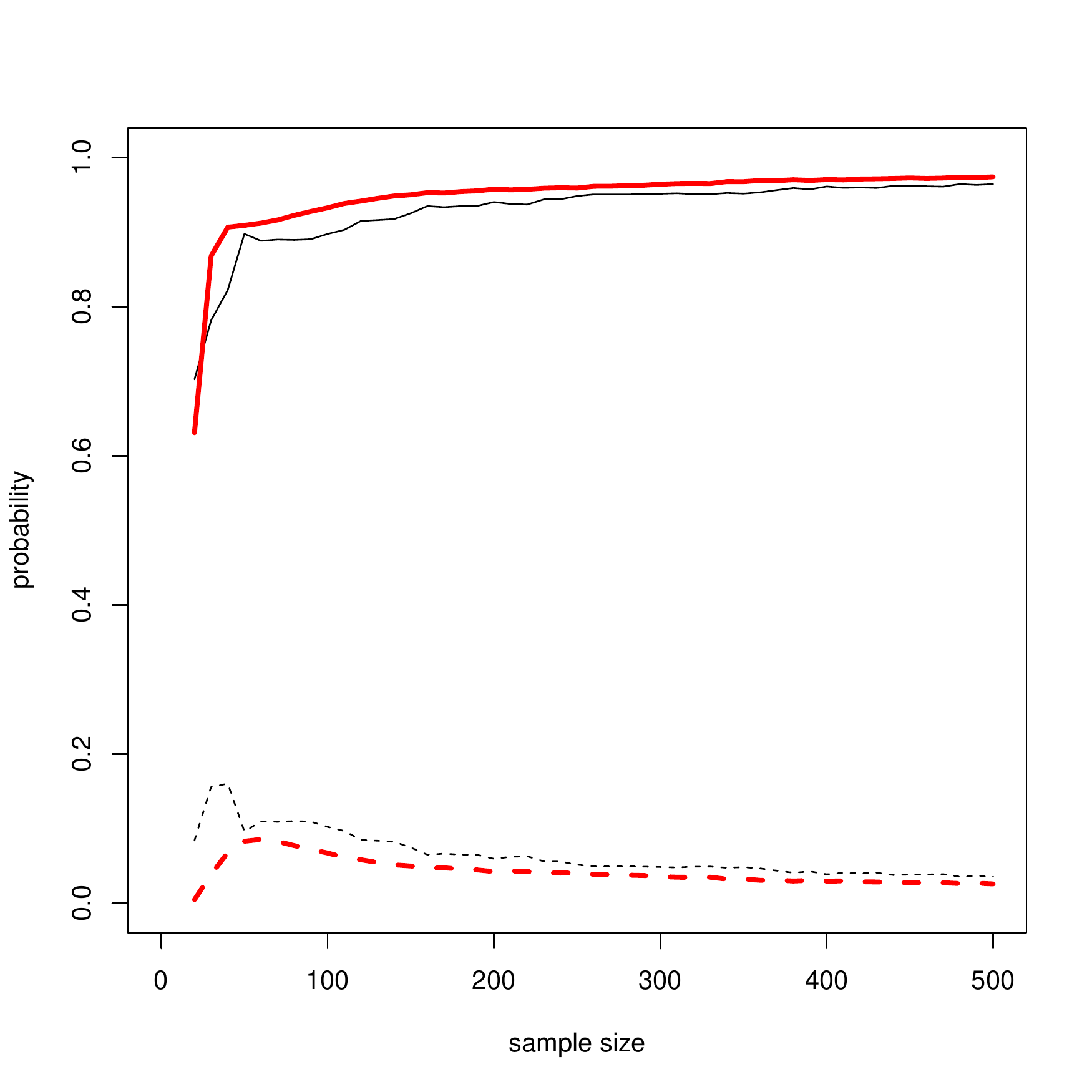} \\
	\centering (a)      \\
	 \centering \includegraphics[scale=0.38]{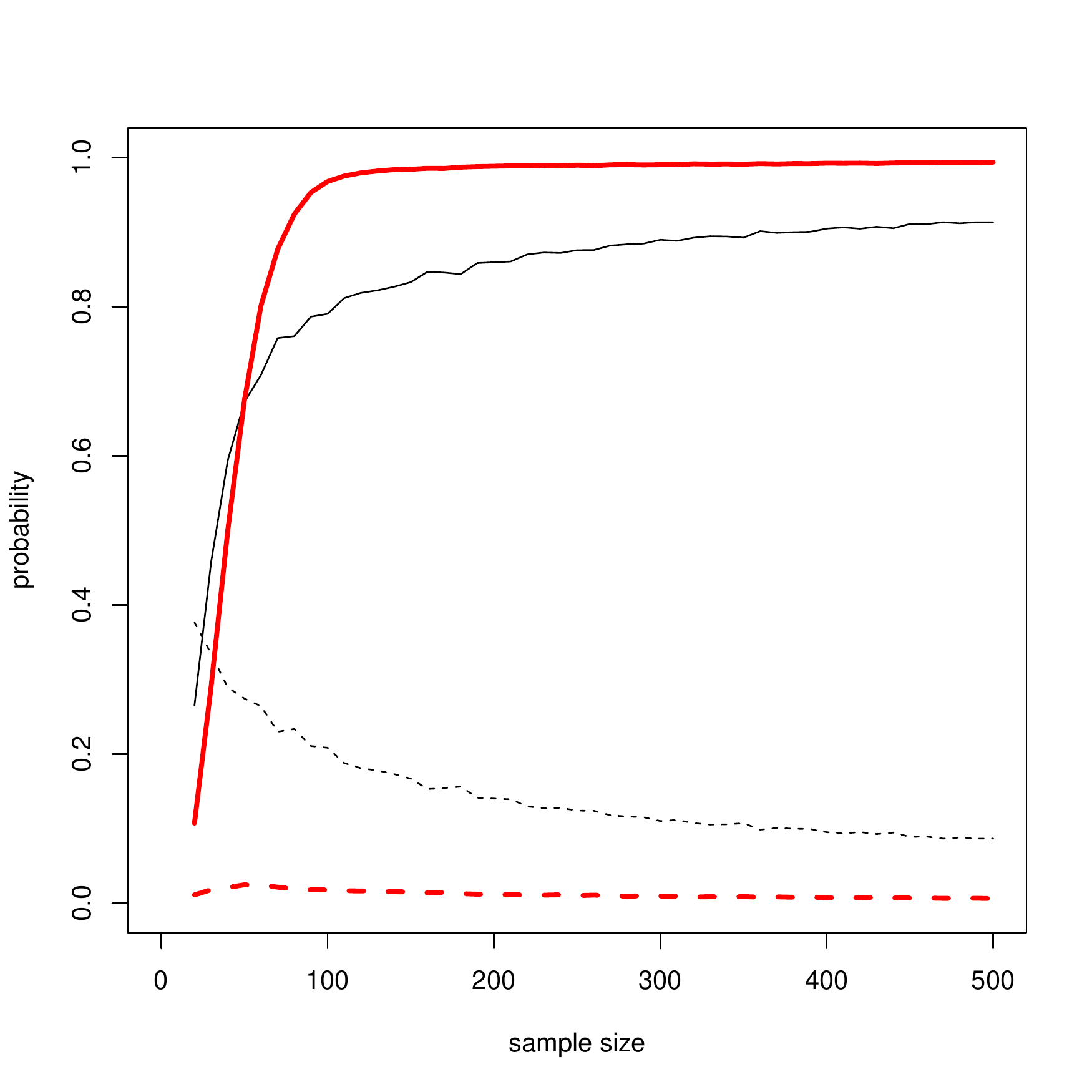} \\
	\centering (c) \\
   \end{minipage}\hfill
   \begin{minipage}{0.49\textwidth}   
      \centering \includegraphics[scale=0.38]{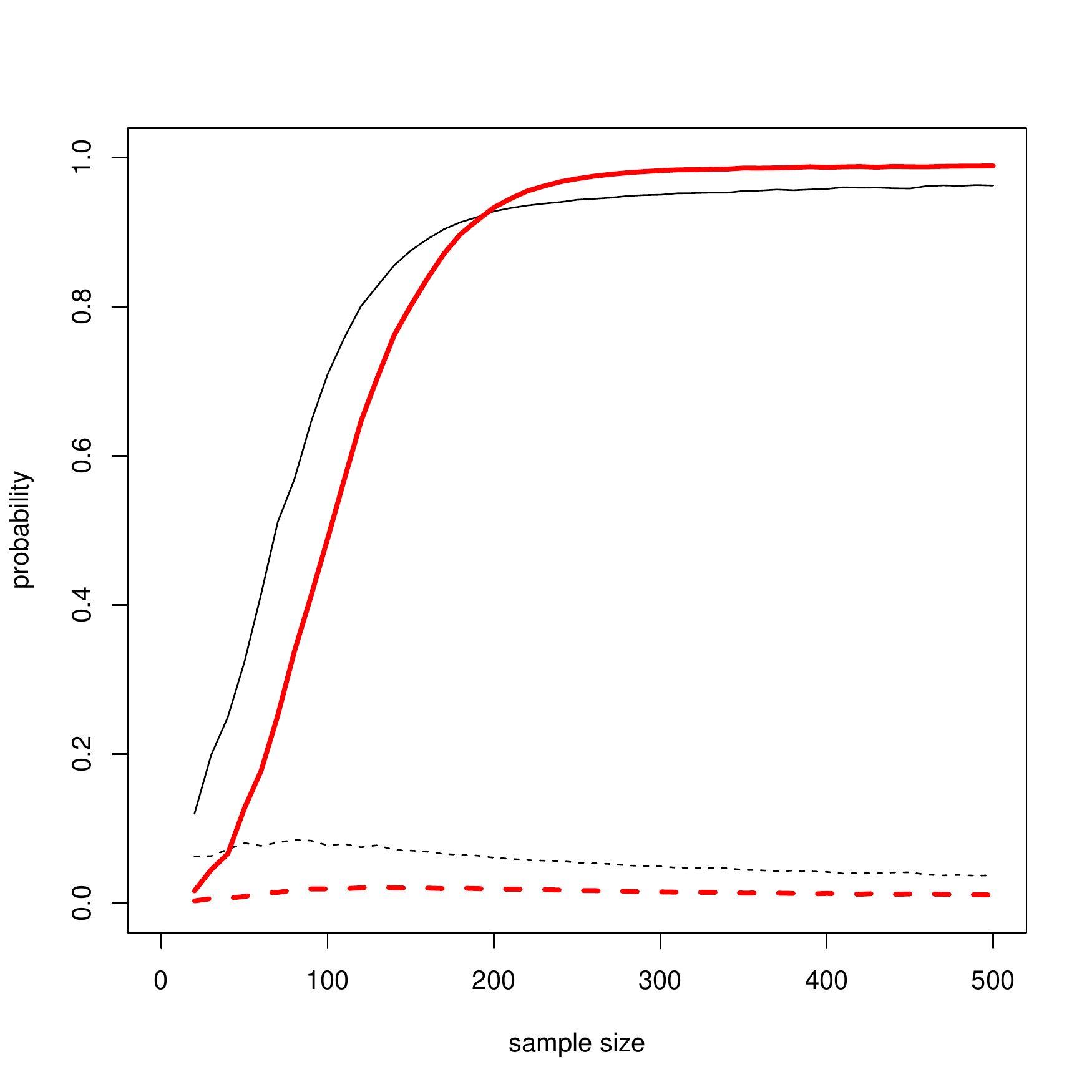} \\
	\centering (b)   \\  
	 \centering \includegraphics[scale=0.38]{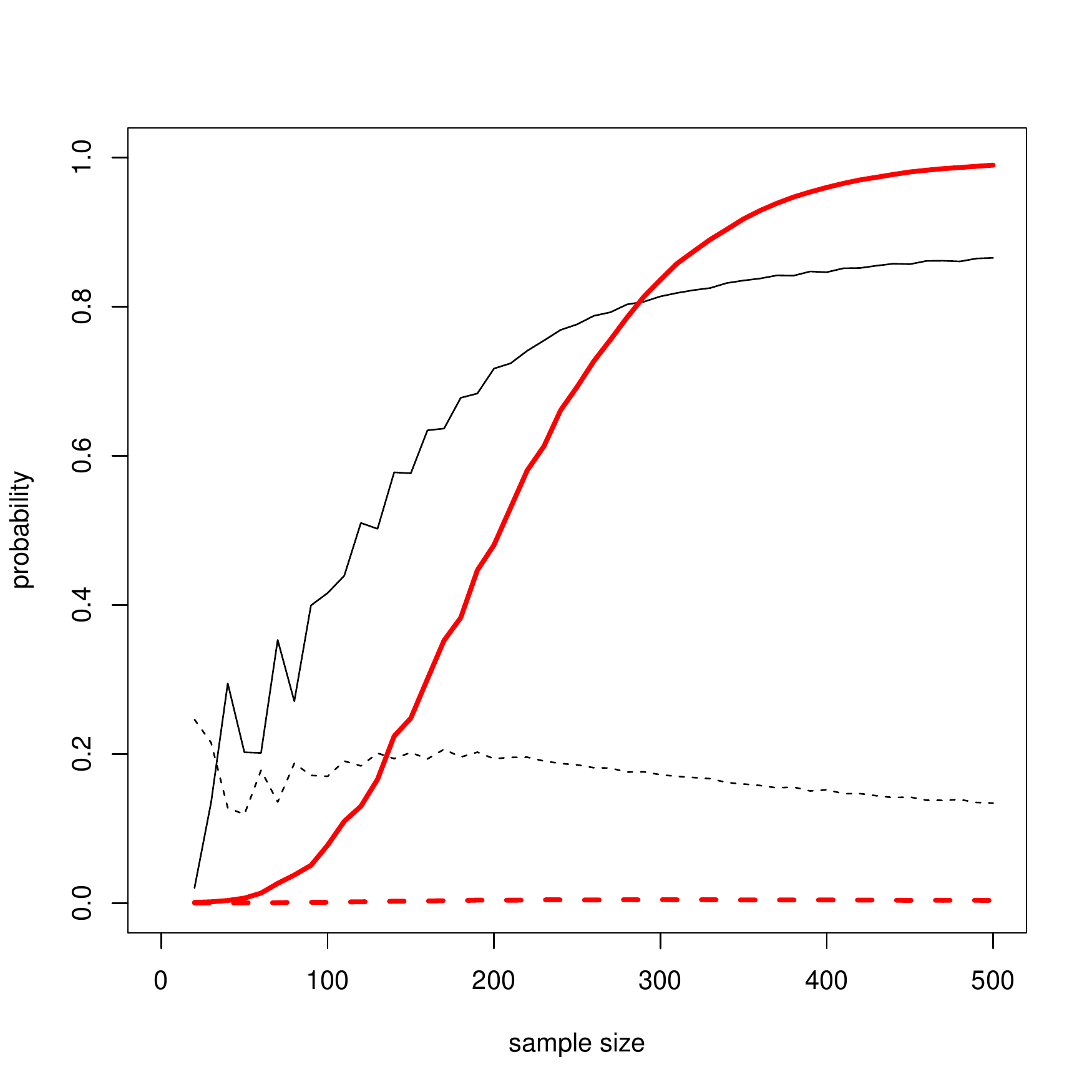} \\
	\centering (d) \\
\end{minipage}
\caption{Probability that the \textsc{bic} criterion (represented in fine black lines) and the proposed approach (represented in bold red lines) select the true number of modes (represented in plain line) and  over-estimate (represented in dotted line)  (a) r=0.3, s=9; (b) r=0.2, s=9; (c) r=0.2, s=18; (d) r=0.1, s=27.}  \label{res::simul}
\end{figure}

In case~(a), modes have a large probability mass and they are easily detected since there are few modalities. Thus, both criteria have the same behavior since they find the true number of modes with a probability close to one even for small samples. 

When the mode probabilities decrease (case~(b)), it is more difficult to identify them. In such case, the \textsc{bic} criterion allows to better find the true number of modes with a moderate overestimation risk, for the small samples (size lower than 150), than the proposed approach which can underestimates the number of modes. When the sample size is larger than 200, the proposed approach obtains better results since it finds a true number of modes almost always while the \textsc{bic} criterion keeps an overestimation risk.

If the number of modalities increases (case~(c)), then the problem becomes harder and the proposed approach also shows its interest since the \textsc{bic} criterion is strongly biased in such case. The \textsc{bic} criterion keeps this bias even for a large data set while the proposed approach almost always finds the true number of modes when the sample size is larger than 100.

Finally, note that in the more complex situations like in case~(d) (few probability mass for the modes and large number of modalities), the proposed approach underestimates the number of modes when the sample size is small then converges to the true mode values when the sample size increases. Note that, in such case, the bias of the \textsc{bic} criterion keeps significant even for a large data set.

Based on this experiment, the proposed criterion seems most relevant since its asymptotic behavior is better than the \textsc{bic} criterion, it never overestimates the mode number and its variability the smaller than the \textsc{bic} criterion.
\subsection{Simulation with well specified model}
\paragraph{Aim} 
During this experiment, we highlight the good behavior of the algorithms (Metropolis-within-Gibbs sampler and \textsc{em} algorithm) for performing the model selection and the estimation of the \textsc{mle}. Thus, data are generated according to \textsc{cmm}, then the model and the \textsc{mle} are estimated. The quality of the estimation is determined by the Kullback-Leibler divergence. We show that this divergence converges to zero when the sample size increases. So, we conclude to the good behavior of both algorithms.

\paragraph{Data generation} A data set of six variables with three modalities is generated according to a bi-component \textsc{cmm} with the following parameters: $\boldsymbol{\sigma}=( \{1 ,2 \} , \{ 3,4\} , \{ 5,6 \} )$, $\ell_{kj}=2$, $\boldsymbol{\pi}=(0.5,0.5)$, $\boldsymbol{\alpha}_{kj}=(0.4,0.4,0.2/7)$,  where the modes are located at different modality crossings for both classes.

\paragraph{Results} For different values of $n=(50,100,200,400,800)$, 100 samples are generated. The Kullback-Leibler divergence is computed between the true and the estimated parameters. Table~\ref{simul1} presents the mean of this divergence.

\begin{table}[h'!]
\begin{center}
\begin{tabular}{rccccc}
\hline $n$ & 50 & 100 & 200 & 400 & 800 \\
\hline mean& 0.656 & 0.117 & 0.061 & 0.028 & 0.015 \\
sd & 0.636 & 0.052 & 0.018 & 0.007 & 0.003 \\
\hline
\end{tabular}
\caption{Mean and standard deviation of the Kullback-Leibler divergence computed between the true parameters of the specified model and the maximum likelihood estimates associated to the model selected by the Metropolis-within-Gibbs algorithm for different sample size.\label{simul1}}
\end{center}
\end{table}

As the Kullback-Leibler divergence converges to zero, when the sample size increases, we claim that the estimated distribution converges to the true one. Thus, we conclude to the good behavior of the estimation algorithm.

\subsection{Simulation with misspecified model}
\paragraph{Aim} 
During this experiment, we underline that the flexibility of \textsc{cmm} allows it to keep good results even if the model is misspecified. Thus, we simulate samples according to a bi-component mixture model where the intra-class dependencies are different for both components. A tuning parameter allows us to modify the strength of the intra-class dependencies and the class overlapping. The results of \textsc{cmm} are compared to those of \textsc{cim}.

\paragraph{Data generation} A data set of size 100 is sampled from the following bi-component mixture model of dimension six

\begin{small}
\begin{equation}
p(\boldsymbol{x};\boldsymbol{\theta})=0.5 \prod_{h=1}^3 p(\boldsymbol{x}^{2h-1},\boldsymbol{x}^{2h};\boldsymbol{\theta}) + 0.5\; p(\boldsymbol{x}^1;\boldsymbol{\theta})p(\boldsymbol{x}^6;\boldsymbol{\theta})\prod_{h=1}^2 p(\boldsymbol{x}^{2h},\boldsymbol{x}^{2h+1};\boldsymbol{\theta}),
\end{equation}
\end{small}
with $p(\boldsymbol{x}^j,\boldsymbol{x}^{j+1};\boldsymbol{\theta})=p(\boldsymbol{x}^j;\boldsymbol{\theta})\big( \lambda \mathds{1}_{\{\boldsymbol{x}^j=\boldsymbol{x}^{j+1}\}} + (1-\lambda) p(\boldsymbol{x}^{j+1};\boldsymbol{\theta}) \big)$  and with $\boldsymbol{x}^j\sim\mathcal{M}_3(1/3,1/3,1/3).$
Thus, when $\lambda=0$, the sample is generated by a uniform distribution and classes are confused. The larger is the tuning parameter $\lambda$, the larger are the intra-class dependencies and the class separation. Note that \textsc{cmm} is not the good model since the conditionally correlated variables are not the same in both classes.

\paragraph{Results} For different values of $\lambda=(0.2,0.4,0.6,0.8)$, 100 samples are generated. The Kullback-Leibler divergence associated to the model with the best number of classes (selected by the \textsc{bic} criterion among $g=1,2,3,4$) is computed. Table~\ref{simul2} presents the results obtained by \textsc{cmm} and \textsc{cim}.\\
\begin{table}[h'!]
\begin{center}
\begin{tabular}{rcccc}
\hline $\lambda$ & 0.2 &0.4 & 0.6 & 0.8 \\
\hline \textsc{cmm} & 0.09 (1.00)  & 0.25 (1.16)  & 0.53 (2.08) & 0.87 (2.10)  \\
	\textsc{cim} & 0.11 (1.00) & 0.27 (1.00) & 1.67 (1.12)  & 5.79 (1.40)  \\
\hline
\end{tabular}
\caption{Kullback-Leibler divergence and mean of the class number obtained by \textsc{cmm} and \textsc{cim}.\label{simul2}}
\end{center}
\end{table}

The larger is $\lambda$, the larger is the Kullback-Leibler divergence for both models. However, the flexibility of \textsc{cmm} allows to keep an acceptable value of the Kullback-Leibler divergence while this divergence grows dramatically faster with \textsc{cim}. Furthermore, when the classes are well separated (large value of $\lambda$), \textsc{cmm} finds more often the true class number than \textsc{cim}.

\section{Applications} \label{sec::appli}
For both applications, the estimation of \textsc{cmm} was performed by the R package \texttt{CoModes}. Both data set are available in \texttt{CoModes} developed by the authors. \ref{tutorial} displays the R code of the second application and can be used as a tutorial of \texttt{CoModes}.

\subsection{Seabirds clustering}
\paragraph{Data} We study a biological data set describing 153 puffins (seabirds) by five plumage and external morphological characteristics presented in Table~\ref{data_birds} \citep{Bre07}. These seabirds are divided into three subspecies
\emph{dichrous} (84 birds), \emph{lherminieri} (34 birds) and \emph{subalaris} (35 birds).

\begin{table}[`!ht']
\begin{center}
\begin{tabular}{rccccccc}
\hline variables&$m_j$ & \multicolumn{4}{c}{modalities}\\
\hline collar & 5 &none & ... & ... $\quad$ ... & continuous \\
eyebrows & 4 &none & ... &  ... &   very pronounced & \\
sub-caudal& 4 & white & black & black and white &  BLACK and white & \\
border & 3& none & ... &  many & & \\
gender &2& male & female & & & \\
\hline
\end{tabular}
\end{center}
\caption{Presentation of the five plumage and external morphological variables describing the puffins. \label{data_birds}}
\end{table}

\paragraph{Experimental settings} The subspecies memberships of the individuals are blinded. For $g=1,\ldots,6$, the \textsc{mle} of \textsc{cim} is obtained by 25 initializations of an \textsc{em} algorithm while 25 chains of 3000 iterations are performed for the model selection of \textsc{cmm} followed by 25 initializations of \textsc{em} algorithm to find the \textsc{mle}.

\paragraph{Results} Table~\ref{bic1} presents the values of the \textsc{bic} criterion for both models and different class numbers. Even if both models select two components, the values of the \textsc{bic} criterion are better for \textsc{cmm} than for \textsc{cim} for all the number of classes. Thus, \textsc{cmm} better fits the data than \textsc{cim}. 
\begin{table}[h!]
\begin{center}
\begin{tabular}{rcccccc}
\hline $g$ & 1 & 2 & 3 & 4 & 5 & 6 \\
\hline \textsc{cmm} & -711 & \textbf{-691} & -701 & -709 & -721 & -727	 \\
\textsc{cim} & -711 & \textbf{-706} & -722 & -745 & -775 & -805 \\
\hline
\end{tabular}
\end{center}
\caption{Values of the \textsc{bic} criterion for different class numbers and for \textsc{cmm} and \textsc{cim}. Boldface indicates the best values of this criterion. \label{bic1}}
\end{table}

According to Table~\ref{confu1} displaying the confusion matrix between the estimated partitions and the subspecies, we claim that the Subalaris are more different than the two other subspecies. Indeed, both models affect all the Subalaris in class 2. If the estimated partitions by both models are similar, we remark that \textsc{cmm} affects less other subspecies in this class than \textsc{cim}.  
\begin{table}[h!]
\begin{center}
\begin{tabular}{rcccc}
\hline & \multicolumn{2}{c}{\textsc{cmm}}& \multicolumn{2}{c}{\textsc{cim}} \\
	& class 1 & class 2 & class 1 & class 2 \\
\hline Dichrous & 52 & 32 & 48 & 36\\
Lherminieri & 23 & 11 & 22 & 12\\
Subalaris & 0 & 35  & 0 & 35\\
\hline
\end{tabular}
\caption{Confusion tables between the subspecies and estimated partition into two classes.\label{confu1}}
\end{center}
\end{table}

Figure~\ref{acmbirds}(a) displays the seabirds on the first correspondence analysis plan and indicates the subspecies. We note that all the Subalaris are in the same location (bottom left) for the first principal correspondence map. We display the partition corresponding to the best model (\textsc{cmm} with two components) in Figure~\ref{acmbirds}(b). Note that, for both model, the first principal correspondence axe allows to define a classification rule.

\begin{figure}[h!]
\begin{center}
\begin{minipage}{0.49\textwidth}
\centering \includegraphics[scale=0.37]{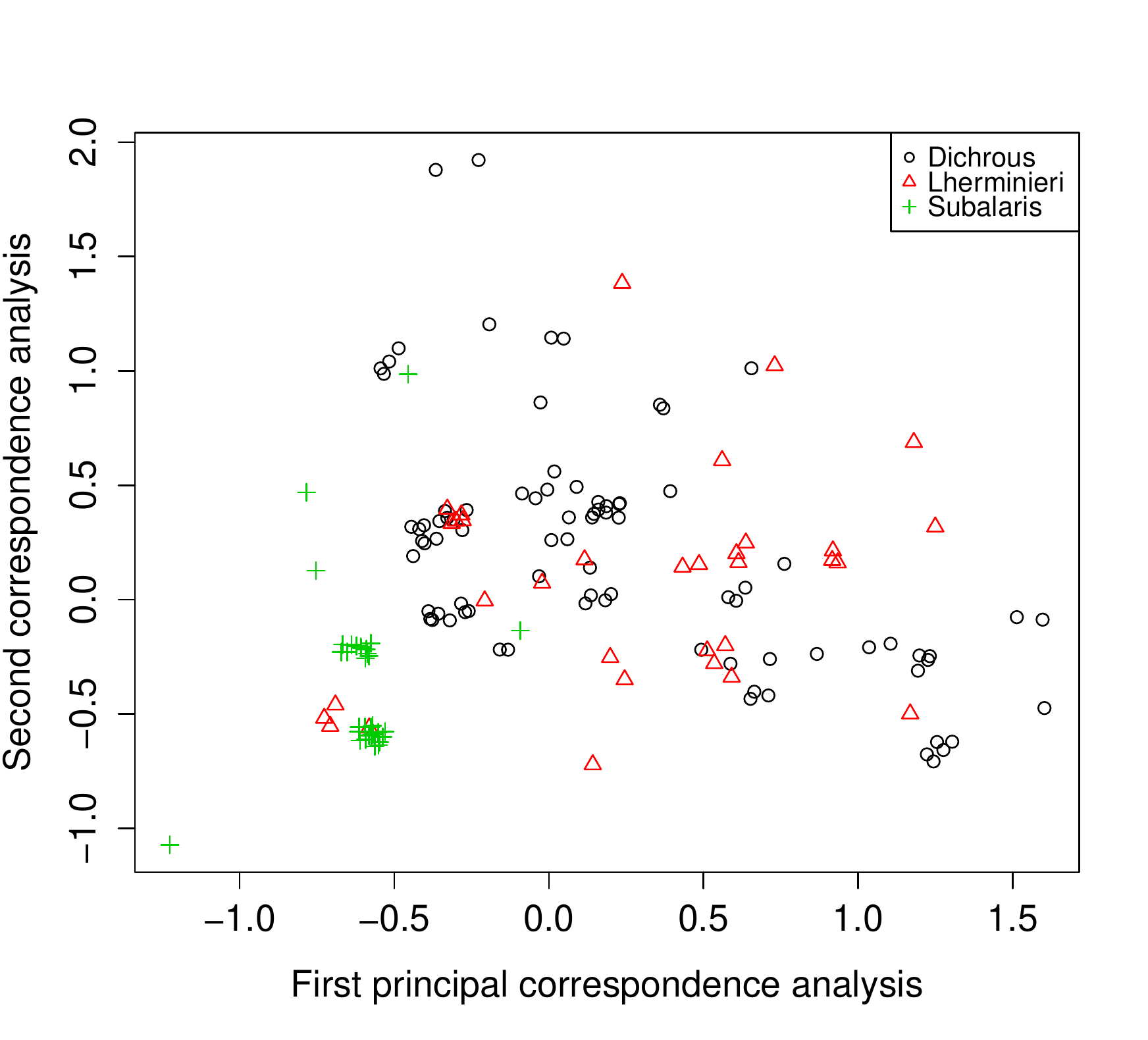} \\
\centering (a)\\
\end{minipage}
\begin{minipage}{0.49\textwidth}
\centering \includegraphics[scale=0.37]{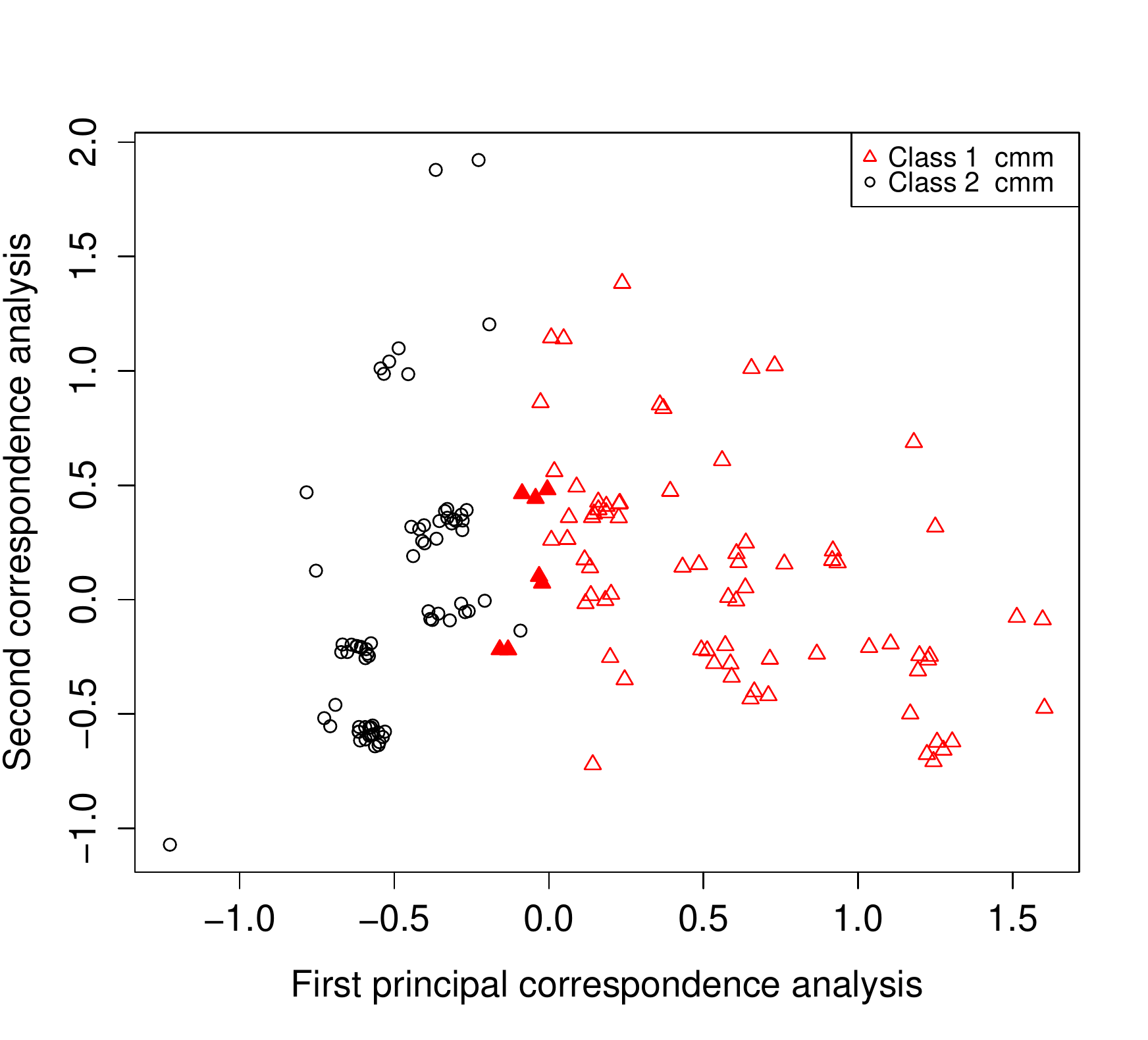} \\
\centering (b)\\
\end{minipage}
\end{center}
\caption{Seabirds on the first principal correspondence analysis map (a) with the subspecies and (b) with the best \textsc{cmm} estimated partition. The bold triangles indicate the individuals affected in class 1 for \textsc{cmm} and in class 2 for \textsc{cim}. An i.i.d. uniform noise on $[0,0.1]$ has be added on
both axes for each individual in order to improve visualization. \label{acmbirds}}
\end{figure}

We now describe the best bi-component model of \textsc{cmm}. Even if the estimated model assumes the conditional independence between the variables, this model is of interest because of its sparsity. Indeed, it is more parsimonious than \textsc{cim} since a small number of modes is estimated as shown by the summary proposed by $\kappa_{kj}$ and $\rho_{kj}$ defined in \eqref{summ} and presented in Table~\ref{summ1}. Thus, the first variables are characterized by few modalities with a high probability. As the variables are conditionally independent, the $\kappa_{kj}$  indicates the number of modalities having a probability upper than the uniform distribution. For example the multinomial distribution of the variable sub-caudal has two modes for both classes (so $\kappa_{kj}=2/3$).
\begin{table}[h!]
\begin{center}
\begin{tabular}{rccccc}
\hline & collar & eyebrows & sub-caudal & border & gender \\
\hline class 1 & 0.75 (0.93) & 0.67 (0.91)& 0.67 (0.88)& 1.00 (1.00)& 1.00 (0.55)\\
class 2 & 0.75 (0.98)& 0.67 (0.77)& 0.67 (0.99)& 0.50 (0.97)& 1.00 (0.57)\\
\hline
\end{tabular}
\end{center}
\caption{Summary of the CMM with three classes: $\kappa_{kj}$ is displayed in plain and $\rho_{kj}$ is displayed in parenthesis.  \label{summ1}}
\end{table}

The maximum likelihood estimates of the component parameters  are displayed in Figure~\ref{birds_param}. Each sub-figure corresponds to a block of variable, thus we note again that the estimated model assumes the conditional independence. For each block of variables, the modality crossings where one mode is estimated for  at least one component are focused. For these modality crossings, we display their cumulated probability masses for each component (the component are identifiable by different colors). These modality crossings are presented by decreasing order of cumulated probability mass.
Note that the mode locations are discriminative since the modality black (resp. white) has a probability of 0.64 (resp. 0.24) for the class 1 while the modality white (resp. BLACK and white) has a probability of 0.94 (resp. 0.05). 

\begin{figure}
\begin{center}
\includegraphics[scale=0.5]{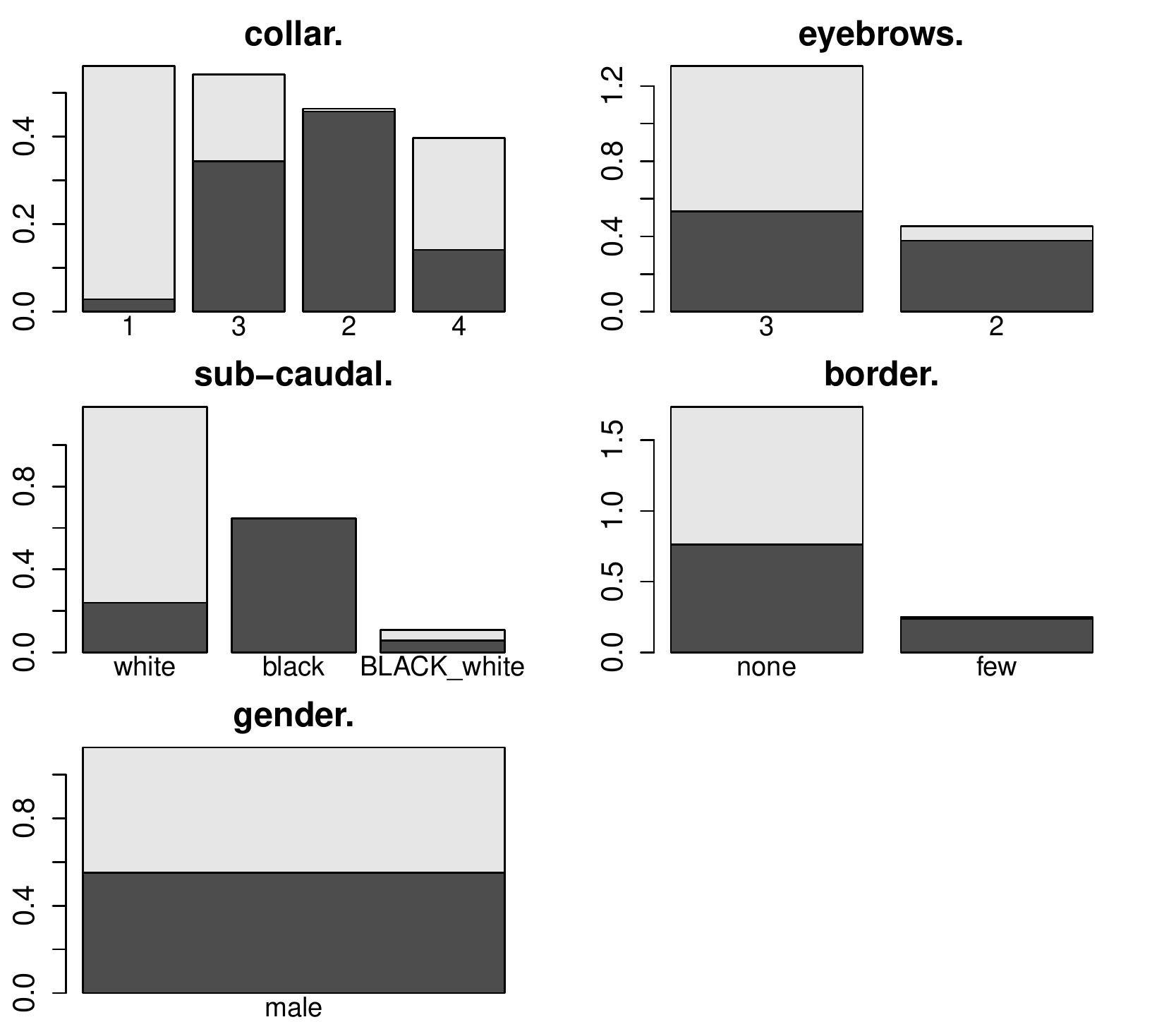} 
\caption{Class parameters of the bi-components \textsc{cmm} estimated on the Seabirds data. The black color (respectively the gray color) corresponds to the probability mass of the modes for the class 1 (respectively to the class 2).\label{birds_param}}
\end{center}
\end{figure}

Finally, the conditional independence assumption seems realistic since the conditional Cramer's V measures, presented in Table~\ref{vcramebirds}, are small.
\begin{table}[h!]
\begin{center}
\begin{minipage}{0.49\textwidth}
\begin{tabular}{ccccc}
1 & 0.14 & 0.15 & 0.23 & 0.21 \\ 
 & 1 & 0.36 & 0.20 & 0.13 \\ 
 &  & 1 & 0.13 & 0.19 \\ 
 &  &  & 1 & 0.01 \\ 
 &  &  &  & 1 \\ 
\end{tabular} \\
\centering (a) Class 1\\
\end{minipage}
\begin{minipage}{0.49\textwidth}
\begin{tabular}{ccccc}
1 & 0.14 & 0.09 & 0.11 & 0.28 \\ 
 & 1 & 0.24 & 0.21 & 0.26 \\ 
 &  & 1 & 0.02 & 0.07 \\ 
 &  &  & 1 & 0.17 \\ 
 &  &  &  & 1 \\ 
\end{tabular} \\
\centering (b) Class 2\\
\end{minipage}
\end{center}
\caption{Matrix of the Cramer's V measures computed according to the estimated classes.\label{vcramebirds}}
\end{table}
We also perform a bootstrap test of the global nullity of the Cramer's V by generating 1000 samples. We obtain a p-value of 0.91, so the conditional independence assumption is validated.

\subsection{Acute inflammations clustering}
A tutorial of the R package \texttt{CoModes} performing the clustering of the Acute inflammations data set is presented in \ref{tutorial}.

\paragraph{Data} We want to cluster 120 patients \citep{Cze03} described by five binary variables (occurrence of nausea (Nau), lumbar pain (Lum), urine pushing (Pus), micturition pains (Mic) and burning of urethra (Bur)) and by one three modalities variables (temperature of the patient (Tem): $T<37C$, $37°C\leq T<38°C$ and $38°C\leq T$). We know that some patients have one of the following diseases of the urinary system: inflammation of urinary bladder and Nephritis of renal pelvis origin.

\paragraph{Experimental conditions} We use the same experimental conditions as the Seabirds clustering.

\paragraph{Results} Table~\ref{bic2} presents the values of the \textsc{bic} criterion for both models and different class numbers. For each class number, the \textsc{bic} criterion value of \textsc{cmm} is better than for \textsc{cim}. Futhermore, \textsc{cmm} selects three classes while \textsc{cim} selects four classes. This phenomenon can be due to the violated conditional independence assumption of \textsc{cim}.
\begin{table}[h!]
\begin{center}
\begin{tabular}{rcccccc}
\hline $g$ & 1 & 2 & 3 & 4 & 5 & 6 \\
\hline \textsc{cmm} & -510 & -351 & \textbf{-338} & -345 & -399 & -401\\
\textsc{cim} & -527 & -478 & -439 & \textbf{-407} & -412 & -418 \\
\hline
\end{tabular}
\end{center}
\caption{Values of the \textsc{bic} criterion for different classes number and for \textsc{cmm} and \textsc{cim}. Boldface indicates the best values of this criterion. \label{bic2}}
\end{table}

Note that the estimated distributions of \textsc{cim} and \textsc{cmm} are different. The obtained partition are also different. Table~\ref{acute_conf} displays the confusion matrices between the best model of \textsc{cmm} and the models of \textsc{cim} with three and four classes. Thus, if 29 individuals constitute a group which is well separated of the other individuals (class 3) for the three models, the other individuals have a class membership determined by the selected model.
\begin{table}[h!]
\begin{center}
\begin{minipage}{0.49\textwidth}
\begin{center}
\begin{tabular}{c|ccc}
 &\multicolumn{3}{c}{\textsc{cmm}}\\ 
 & c1 & c2 & c3 \\ 
\hline \textsc{cim} c1 & 40 & 0 & 0 \\ 
\textsc{cim} c2 & 10 & 41 & 0 \\ 
\textsc{cim} c3 & 0 & 0 & 29 \\ 
\end{tabular} 
\end{center}
\end{minipage}
\begin{minipage}{0.49\textwidth}
\begin{center}
\begin{tabular}{c|ccc}
 &\multicolumn{3}{c}{\textsc{cmm}}\\ 
 & c1 & c2 & c3 \\ 
\hline \textsc{cim} c1 & 40 & 0 & 0 \\ 
\textsc{cim} c2 & 10 & 20 & 0 \\ 
\textsc{cim} c3 & 0 & 21 & 0 \\ 
\textsc{cim} c4 & 0 & 0 & 29 \\ 
\end{tabular} 
\end{center}
\end{minipage}
\end{center}
\caption{Confusion matrices between the best model of \textsc{cmm} and the models of \textsc{cim} with three and four classes. \label{acute_conf}}
\end{table}

Figure~\ref{diagno} displays the individuals on the 1-5 principal correspondence analysis map where the estimated classes are well separated. 
\begin{figure}[h!]
\begin{center}
\centering \includegraphics[scale=0.3]{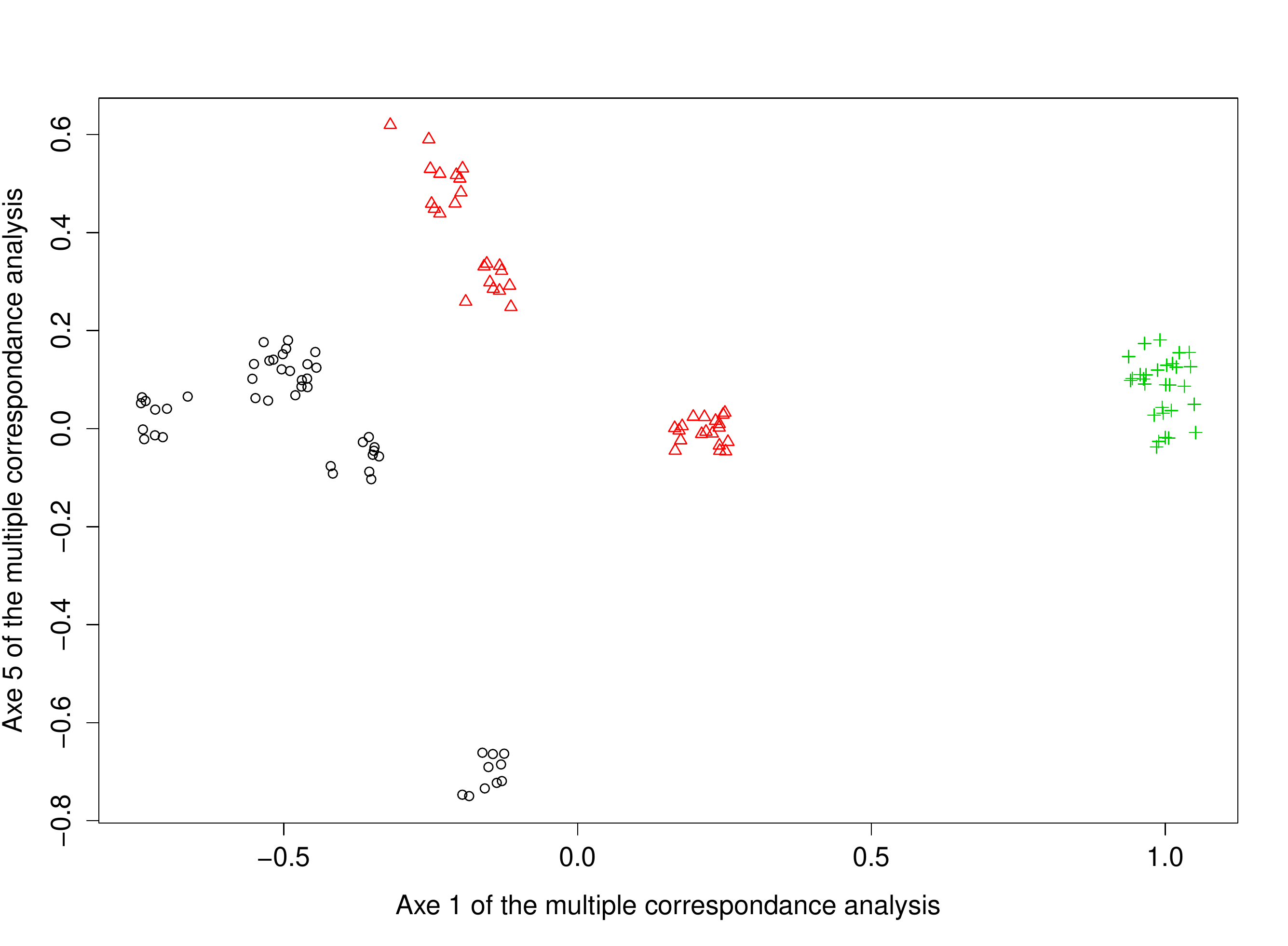} 
\end{center}
\caption{Individuals on the 1-5 principal correspondence analysis map with the best \textsc{cmm} estimated partition. An i.i.d. uniform noise on $[0,0.1]$ has be added on
both axes for each individual in order to improve visualization. Colors and symbols indicate the class membership.\label{diagno}}
\end{figure}

The model \textsc{cmm} with three classes has the following repartition of the variables into blocks: $\boldsymbol{\sigma}=(\{\text{Tmp, Pus, Mic, Bur}\},\{\text{Nau}\},\{\text{Lum}\})$. As shown by the summary $\rho_{kj}$ and $\kappa_{kj}$ displayed in Table~\ref{summ2}, the three classes are concentrated in few modality crossings for the block one and in one location with a probability close to one for the two other blocks.
\begin{table}[h!]
\begin{center}
\begin{tabular}{rccc}
\hline	& Tmp, Nau, Lum, Mic & Pus & Bur\\
\hline Class 1 & 0.41 (1.00)& 1.00 (1.00)& 1.00 (0.99)\\
Class 2 & 0.33 (0.99)& 1.00 (1.00)& 1.00 (1.00)\\
Class 3 & 0.25 (0.99)& 1.00 (1.00)& 1.00 (1.00)\\
\hline
\end{tabular}\\
\end{center}
\caption{Summary of the CMM with three classes: $\kappa_{kj}$ is displayed in plain and $\rho_{kj}$ is displayed in parenthesis.  \label{summ2}}
\end{table}

The following class interpretation is based on the class parameters displayed by Figure~\ref{acute_param}. Note  that the variables \emph{urine pushing} and \emph{burning of urethra} are the most discriminative ones.
\begin{itemize}
\item The majority class (42\%) groups individuals having no nausea and no lumber pain.
\item The second class (34\%) groups individuals having no nausea but lumber pain.
\item The third class (24\%) groups individuals having nausea and lumber pain. Furthermore, these individuals have some fiever and micturition pain.
\end{itemize}

\begin{figure}[h!]
\begin{center}
\includegraphics[scale=0.58]{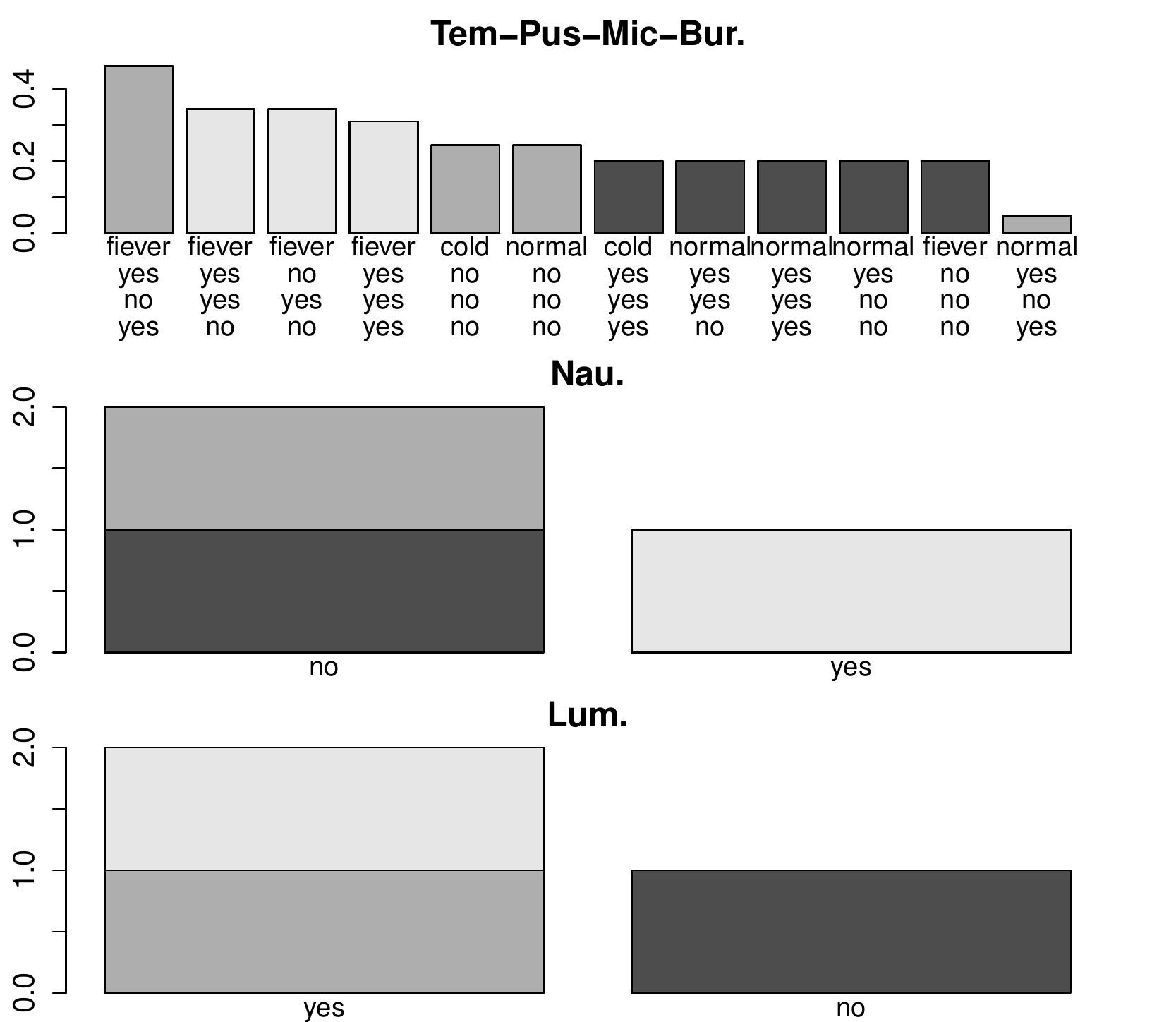} 
\caption{Estimated parameters of the tri-component \textsc{cmm} displayed by the barplot function of the package \texttt{CoModes}. Black color corresponds to class 1, black gray color corresponds to class 2 and pale gray color corresponds to class 3.}\label{acute_param}
\end{center}
\end{figure}

\section{Conclusion} \label{sec::concl}
In this article, we have presented a new mixture model (\textsc{cmm}) to cluster categorical data. Its strength is to relax the conditional independence assumption and to stay parsimonious. A summary of the distribution is given by $\kappa_{kj}$ and $\rho_{kj}$ while each class can be summarized by the mode locations. As shown by the Seabirds application, \textsc{cmm} can improve the results of the classical latent class model even if the conditional independence assumption is true, thank's to its sparsity. 

The combinatorial problems of the block detection and of the modes number selection is solved by a Metropolis-within-Gibbs algorithm and use the computation of the integrated complete-data likelihood. Thus, this approach can be used to select the interactions of the log-linear mixture model per block. The R package \texttt{CoModes} allows to perform the model selection and the parameter's estimation. Both data sets presented in this article are included in this package. To efficiently reduce the computing time, the functions of this package will be soon implemented in C++.

However, the model is hardly estimated if the data set has a large number of variables. Some constraints on the block variables repartition could also be added (for instance the number of variables into blocks could be limit at three variables). Another solution could be to estimate the model by a forward/backward strategy but it is know that these method are sub-optimal.

Finally, we imposed the equality of the repartition of the variables into blocks for all the classes. This property allows us to prove the generic identifiability of \textsc{cmm}. This lack of flexibility is counterbalanced by flexible block distribution. However, one could try to relax the class-equality of $\boldsymbol{\sigma}$ with the model no-identifiability risk.
\bibliography{biblio}

\bibliographystyle{elsarticle-num}

\appendix

\section{Generic identifiability of \textsc{CMM} \label{ident}}
We demonstrate that \textsc{cmm} is generically identifiable, \emph{i.e.} the parameter space where the model is not identifiable has a Lebesgue measure equal to zero. To do this, we adapt the demonstration of the generic identifiability of \textsc{cmm} given by \citep{All09} and based on the Kruskal theorem \citep{Kru77,Kru76}. By using the conditional independence between the blocks of variables, we present a sufficient condition for the generic identifiability of \textsc{cmm} which is a relation between the class number and the mode number.
As this demonstration is just an adaptation of the proof given by \citep{All09}, some technical details are not reminded here. The demonstration is cut into three steps: we start by a reminder of the Kruskal results for the three-way tables, then we show the generic identifiability of \textsc{cmm} with three blocks of variables and we finish by an extension to \textsc{cmm} with more than three blocks.

\paragraph{Kruskal results} For a matrix $M$, the Kruskal rank of $M$, denoted by $\text{rank}_K\; M$ is the largest number $I$ such that every set of $I$ rows of $M$ are linearly independent.\\
T\textsc{heorem 1} (Kruskal \citep{Kru77,Kru76}). Let $I_j=\text{rank}_K\; M_j$. If
\begin{equation*}
I_1 + I_2 +I_3 \geq 2g+2,
\end{equation*}
then the tensor $[M_1, M_2, M_3]$ uniquely determines the $M_j$, up to simultaneous permutation and rescaling rows.

\paragraph{Generic identifiability of \textsc{CMM} with three blocks}
Let  $k_0=\underset{k}{\text{argmin }} \ell_{kj}$ and the matrix $M_j$ where 
\begin{equation}
M_j(k,h)=\alpha_{kj\tau_{k_0j}(h)}.
\end{equation}
By denoting by $\xi_j=\underset{k}{\min }\; \ell_{kj} +1$, generically, we have
\begin{equation*}
\text{rank}_K \; M_j=\min(g,\xi_j).
\end{equation*}
C\textsc{orollary 1} The parameters of \textsc{cmm} with three blocs are generically identifiable, up to label swapping, provided:
\begin{equation*}
\min(g,\xi_1) + \min(g,\xi_2) + \min(g,\xi_3) \geq 2g+2.
\end{equation*}

\paragraph{Generic identifiability of \textsc{cmm} with more than three blocks}
In the same way that \citep{All09}, we generalize the result with $d$ blocks by observing that $d$ blocks of categorical variables can be combined into three categorical variables. Thus, we can apply the Kruskal theorem.\\
C\textsc{orollary 2} We consider a \textsc{cmm} with $d$ blocks where $d\leq 3$. If there exists a tri-partition of the set $\{1,\ldots,d\}$ into three disjoint non empty subsets $S_1$, $S_2$ and $S_3$, such that $\gamma_i=\prod_{j\in S_i}\xi_j$ with
\begin{equation}
\min(g,\gamma_1) + \min(g,\gamma_2) + \min(g,\gamma_3) \geq 2g+2 ,
\end{equation}
then the model parameters are generically identifiable up to label swapping.

\section{Proof of Proposition 1} \label{proofs}
In this Section, a proof of Proposition 1 is given. We firstly define a new parametrization of the block distribution facilitating the integrate complete-data likelihood computation. We secondly define the prior distribution of the new block parametrization according to the other parametrization. Thirdly, we underline the relation between the embedded models. We conclude by the integrate complete-date likelihood computation, which is the target result.

\subsection{New parametrization of the block distribution}
Without loss of generality, we assume that the elements of $\boldsymbol{\delta}_{kj}$ are ordered by decreasing values of the probability mass associated to them and we introduce the new parametrization of $\boldsymbol{a}_{kj}$ denoted $\boldsymbol{\varepsilon}_{kj}$ where $\boldsymbol{\varepsilon}_{kj} \in \mathcal{E}_{kj}=\left[\frac{1}{m_j};1\right]\times,\ldots,\times \left[\frac{1}{\text{m}_j - \ell_{kj}};1\right]$ and where $\varepsilon_{kjh}$ is defined by
\begin{align}
\varepsilon_{kjh} =& \left\{ \begin{array}{rl}
a_{kj\delta_{kjh}} & \text{if } h=1\\
\frac{a_{kj\delta_{kjh}}}{\prod_{h'=1}^{h-1} (1-\varepsilon_{kjh'})} & \text{otherwise}.
\end{array}
\right. \nonumber
\end{align}

\begin{lem}
The conditional probability of $\textbf{X}^j$ is
\begin{align}
p(\textbf{X}^j|\textbf{Z},\ell_{kj},\tilde{\boldsymbol{\delta}}_{kj},\boldsymbol{\varepsilon}_{kj})=
\prod_{h=1}^{\ell_{kj}} (\varepsilon_{kjh})^{\text{n}_{kj(h)}} (1-\varepsilon_{kjh})^{\bar{\text{n}}_{kj}^h},
\end{align}
\end{lem}

\begin{proof}
\begin{align*}
p(\textbf{X}^j|\textbf{Z},\ell_{kj},\tilde{\boldsymbol{\delta}}_{kj},\boldsymbol{\varepsilon}_{kj})
&=p(\textbf{X}^j|\textbf{Z},\ell_{kj},\boldsymbol{\alpha}_{kj})\\
&=\prod_{h=1}^{\text{m}_j} (\alpha_{kjh})^{\text{n}_{kjh}} \\
&=\left[ \prod_{h=1}^{\ell_{kj}} (\alpha_{kj(h)})^{\text{n}_{kj(h)}} \right] \alpha_{kj(\ell_{kj}+1)}^{\bar{n}_{kj}^{\ell_{kj}}}\\
&=\varepsilon_{kj1}^{n_{kj(1)}} \prod_{h=2}^{\ell_{kj}} \left[ \varepsilon_{kjh}^{\text{n}_{kj(h)}} \left(\prod_{h'=1}^{h-1} (1-\varepsilon_{kjh})^{n_{kj(h)}} \right) \right] \prod_{h=1}^{\ell_{kj}} (1-\varepsilon_{kjh})^{\bar{n}_{kj}^{\ell_{kj}}} \\
&=\prod_{h=1}^{\ell_{kj}} (\varepsilon_{kjh})^{n_{kj(h)}} (1-\varepsilon_{kjh})^{\bar{n}_{kj}^h}.
\end{align*}
\end{proof}

\subsection{Prior distribution}
\begin{lem}
The  prior distribution of $\boldsymbol{\varepsilon}_{kj}$ is
\begin{equation}
p(\boldsymbol{\varepsilon}_{kj}|\boldsymbol{\omega},\boldsymbol{\delta}_{kj})=\frac{\text{m}_j}{\text{m}_j - \ell_{kj}}.
\end{equation}
\end{lem}

\begin{proof}
We remind that $\boldsymbol{a}_{kj}|\boldsymbol{\omega} \sim D_{\ell_{kj}+1}^t \Big(1,\ldots,1;\text{m}_j\Big)$ and that
\begin{equation}
p(\boldsymbol{a}_{kj},\boldsymbol{\delta}_{kj}|\boldsymbol{\omega})=
p(\boldsymbol{\alpha}|\boldsymbol{\omega})=
p(\boldsymbol{\varepsilon}_{kj},\boldsymbol{\delta}_{kj}|\boldsymbol{\omega}).
\end{equation}
So, we deduce the pdf of the prior distribution of $\boldsymbol{\varepsilon}_{kj}$
\begin{equation}
p(\boldsymbol{\varepsilon}_{kj}|\boldsymbol{\delta}_{kj},\boldsymbol{\omega})= \frac{
\prod_{h=1}^{\ell_{kj}} (\varepsilon_{kjh})^{\gamma_{kjh}-1} (1-\varepsilon_{kjh})^{\sum_{h'=h+1}^{\ell_{kj}+1} (\gamma_{kjh'}-1)}
}{
\int_{\varepsilon_{kj} \in \mathcal{E}_{kj}} \prod_{h=1}^{\ell_{kj}} (\varepsilon_{kjh})^{\gamma_{kjh}-1} (1-\varepsilon_{kjh})^{\sum_{h'=h+1}^{\ell_{kj}+1} (\gamma_{kjh'}-1)} d\varepsilon_{kj}
}.
\end{equation}
Thus, each $\varepsilon_{kjh}$ follows a truncated Beta distribution on the parameters space $\left[\frac{1}{\text{m}_j-h+1},1\right]$ denoted by $\mathcal{B}e(\gamma_{kjh},\sum_{h'=h+1}^{\ell_{kj}+1} (\gamma_{kjh'}-1)+1)$. To assure the positivity of the parameters of the truncated Beta distributions, we put $\gamma_{kjh}=1$, so 
\begin{equation}
p(\boldsymbol{\varepsilon}_{kj}|\boldsymbol{\delta}_{kj},\boldsymbol{\omega})=\frac{\text{m}_j}{\text{m}_j - \ell_{kj}}.
\end{equation}
\end{proof}

\subsection{Relation between embedded models}
\begin{lem}
Let the model with $\ell_{kj}^{\ominus}$ modes and the parameters $(\tilde{\boldsymbol{\delta}}_{kj}^{\ominus},\boldsymbol{\varepsilon}_{kj}^{\ominus})$ and let the model with $\ell_{kj}$ modes and the parameters $(\tilde{\boldsymbol{\delta}}_{kj},\boldsymbol{\varepsilon}_{kj})$. Both modes are defined as such that $\ell_{kj}^{\ominus}=\ell_{kj}-1$, that the  $\ell_{kj}^{\ominus}$ modes having the largest probabilities have the same locations ($\forall h\in \boldsymbol{\delta}_{kj}^{\ominus},\; h\in \boldsymbol{\delta}_{kj}$) and the same probability masses $(\varepsilon_{kjh}^{\ominus}=\varepsilon_{kjh},\; h<\ell_{kj})$.
These embedded models follow this relation
\begin{equation}
\frac{p(\textbf{X}^j|\textbf{Z},\ell_{kj},\tilde{\boldsymbol{\delta}}_{kj},\boldsymbol{\varepsilon}_{kj})}{p(\textbf{X}^j|\textbf{Z},\ell_{kj}^{\ominus},\tilde{\boldsymbol{\delta}}_{kj}^{\ominus},\boldsymbol{\varepsilon}_{kj}^{\ominus})}
=\frac{ (\text{m}_j - \ell_{kj} +1)^{\bar{\text{n}}_{kj}^{\ell_{kj}-1} -1}  }{ (\text{m}_j - \ell_{kj})^{\bar{\text{n}}_{kj}^{\ell_{kj}}} } (\varepsilon_{\ell_{kj}})^{\text{n}_{kj(\ell_{kj})}} (1-\varepsilon_{\ell_{kj}})^{\bar{\text{n}}_{kj}^{\ell_{kj}}}  . 
\end{equation}
\end{lem}

\begin{proof}
We start by the following relation
\begin{equation}
\frac{p(\textbf{X}^j|\textbf{Z},\ell_{kj},\boldsymbol{\alpha}_{kj})}{p(\textbf{X}^j|\textbf{Z},\ell_{kj}^{\ominus},\boldsymbol{\alpha}_{kj}^{\ominus})}
= \frac{
\alpha_{kj\ell_{kj}}^{\text{n}_{kj(\ell_{kj})}}(\alpha_{kj\ell_{kj}+1})^{\bar{\text{n}}_{kj}^{\ell_{kj}}}
}{
\alpha_{kj\ell_{kj}}^{\ominus \bar{\text{n}}_{kj}^{\ell_{kj}-1}}
}.
\end{equation}
 Note that, $\varepsilon_{kjh}=\varepsilon_{kjh}^{\ominus}$ when ($h=1,\ldots,\ell_{kj}-1$), since $\alpha_{kj(h)}=\alpha_{kj(h)}^{\ominus}$ and $\tilde{\tau}_{\ell_{kj}}(h)=\tilde{\tau}_{\ell_{kj}-1}(h)$ when ($h=1,\ldots,\ell_{kj}-1$). Then, by using the reparamatrization in $\boldsymbol{\varepsilon}_{kj}$, the proof is completed.
\end{proof}

\subsection{Integrated complete-data likelihood}
The integrated complete-data likelihood is finally approximated, by neglecting the sum over the discrete parameters of the modes locations and by performing the exact computation on the continuous parameters, by
\begin{equation}
p(\textbf{X}^{j}|\textbf{Z},\ell_{kj}) \approx \left( \frac{1}{m_j-\ell_{kj}} \right)^{\bar{\text{n}}_{kj}^{\ell_{kj}}}  \prod_{h=1}^{\ell_{kj}} \frac{Bi\left(\frac{1}{\text{m}_j-h+1};\text{n}_{kj(h)}+1;\bar{\text{n}}_{kj}^{h}+1\right)}{\text{m}_j- h},
\end{equation}
where $Bi(x;a,b)=B(1;a,b)-B(x;a,b)$, $B(x;a,b)$ being the incomplete beta function defined by $B(x;a,b)=\int_0^x w^a(1-w)^b dw$. From the previous expression, its is straightforward to obtain $p(\textbf{X},\textbf{Z}|\boldsymbol{\omega})$.

\begin{proof}[Proof of Proposition 1]
 If, for the model with $\ell_{kj}-1$ modes, the best modes locations are known and given by $\tilde{\boldsymbol{\delta}}_{kj}^{\ominus}$ then the conditional probability of $\textbf{X}^j$ for a model with $\ell_{kj}$ modes is
\begin{equation}
p(\textbf{X}^j|\textbf{Z},\ell_{kj},\tilde{\boldsymbol{\delta}}_{kj}^{\ominus},\boldsymbol{\varepsilon}_{kj})=
\frac{1}{\text{m}_j - \ell_{kj} +1 } 
\sum_{\tau \in \{1,\ldots,\text{m}_j\}\setminus \{\tilde{\boldsymbol{\delta}}_{kj}^{\ominus}\}}
 p(\textbf{X}^j|\textbf{Z},\ell_{kj},\{\tilde{\boldsymbol{\delta}}_{kj}^{\ominus},\tau\},\boldsymbol{\alpha}_{kj}^{\ominus},\boldsymbol{\varepsilon}_{kj}),
\end{equation}
 Thus, by approximating this sum by its maximum element, we obtain that
\begin{equation}
p(\textbf{X}^j|\textbf{Z},\ell_{kj},\tilde{\boldsymbol{\delta}}_{kj}^{\ominus},\boldsymbol{\varepsilon}_{kj})\approx \frac{1}{\text{m}_j - \ell_{kj} +1 } p(\textbf{X}^j|\textbf{Z},\ell_{kj},\tilde{\boldsymbol{\delta}}_{kj},\boldsymbol{\alpha}_{kj}^{\ominus},\boldsymbol{\varepsilon}_{kj}).
\end{equation}
By using the proposition 3, we obtain that:
\begin{equation}
\frac{p(\textbf{X}^j|\textbf{Z},\ell_{kj},\tilde{\boldsymbol{\delta}}_{kj}^{\ominus},\boldsymbol{\varepsilon}_{kj})}{p(\textbf{X}^j|\textbf{Z},\ell_{kj}^{\ominus},\tilde{\boldsymbol{\delta}}_{kj}^{\ominus},\boldsymbol{\varepsilon}_{kj}^{\ominus})}
\approx
\frac{ (\text{m}_j - \ell_{kj} +1)^{\bar{\text{n}}_{kj}^{\ell_{kj}-1}-1}  }{ (\text{m}_j - \ell_{kj})^{\bar{\text{n}}_{kj}^{\ell_{kj}}} } (\varepsilon_{\ell_{kj}})^{\text{n}_{kj(\ell_{kj})}} (1-\varepsilon_{\ell_{kj}})^{\bar{\text{n}}_{kj}^{\ell_{kj}}}  . 
\end{equation}
 As $p(\textbf{X}^{j}|\textbf{Z},\ell_{kj}=0)=(m_j)^{-n_k}$, by applying recursively the previous expression, we obtain that
\begin{equation}
p(\textbf{X}^{j}|\textbf{Z},\ell_{kj},\boldsymbol{\varepsilon}_{kj})\approx
\left( \frac{1}{\text{m}_j-\ell_{kj}} \right)^{\bar{\text{n}}_{kj}^{\ell_{kj}}} \prod_{h=1}^{\ell_{kj}} \frac{(\varepsilon_{kjh})^{\text{n}_{kj(h)}} ( 1 - \varepsilon_{kjh})^{\bar{\text{n}}_{kj}^h}}{\text{m}_j - h +1} .
\end{equation}
\end{proof}

\section{Acute inflammation data set clustering with the R package \texttt{CoModes}} \label{tutorial}
$\#$ Package loading\\
$>$ \texttt{require(CoModes)}\bigskip \\
$\#$ Loading of the data set Acute\\
$>$ \texttt{data(acute)}\bigskip \\
$\#$ Discretization of the first variable to obtain categorical variables\\
$>$ \texttt{acute[acute[,1]}$<$\texttt{37,1]} $<$-\texttt{ 1}\\
$>$ \texttt{acute[acute[,1]}$>$\texttt{38,1]} $<$-\texttt{ 3}\\
$>$ \texttt{acute[acute[,1]}$>$\texttt{3,1]} $<$-\texttt{ 2}\\
$>$ \texttt{acute[,1] <- factor(acute[,1],levels=1:3,}
\flushright{\texttt{labels=c("cold","normal","fiever"))}} \bigskip \\
\flushleft{$\#$ Model selection and parameter estimation of CMM}\\
$>$\texttt{res.CoModes $<$- CoModescluster(acute[,1:6],2)}\bigskip \\
$\#$ Summary of the model\\
$\#$ Table like Table~\ref{summ2}\\
$>$\texttt{summary(res.CoModes)}\bigskip \\
$\#$ Plot of the parameters\\
$\#$ Barplot like Figure~\ref{acute_param}\\
$>$\texttt{barplot(res.CoModes)}\bigskip \\
$\#$ Plot of the individuals in a multiple correspondence analysis map, \\
$\#$ the colors and symbols indicate the class membership estimated by \textsc{cmm}.\\
$\#$ Plot like Figure~\ref{diagno}\\
$>$\texttt{plot(res.CoModes,c(1,5))}\\

\end{document}